\documentclass[sigconf]{aamas}

\usepackage{balance} % for balancing columns on the final page

\usepackage{pifont}
\usepackage{hyperref}
\usepackage{bm}
\usepackage{enumitem}
\usepackage{adjustbox}
\usepackage{subcaption}

\usepackage{graphicx}
\usepackage{enumitem}
\usepackage[american]{babel}
\usepackage{amsfonts}
\usepackage{algorithm}
\usepackage[noend]{algpseudocode}
\usepackage{multirow}
\newtheorem{Claim}{Claim}
\newtheorem{Theorem}{Theorem}
\newtheorem{Remark}{Remark}
\newtheorem{Example}{Example}
\renewcommand{\paragraph}[1]{\smallskip\noindent\textbf{#1}}

    \usepackage[many]{tcolorbox}    
\newtcolorbox{mybox}[1]{%
    tikznode boxed title,
    enhanced,
    arc=0mm,
    interior style={white},
    attach boxed title to top center= {yshift=-\tcboxedtitleheight/2},
    fonttitle=\bfseries,
    colbacktitle=white,coltitle=black,
    boxed title style={size=normal,colframe=white,boxrule=0pt},
    title={#1}}
\usepackage{mathtools} % amsmath with fixes and additions
\usepackage{booktabs} % commands to create good-looking tables
\usepackage{tikz} % nice language for creating drawings and diagrams

\usepackage{xr} 
\makeatletter
\newcommand*{\addFileDependency}[1]{% argument=file name and extension
  \typeout{(#1)}
  \@addtofilelist{#1}
  \IfFileExists{#1}{}{\typeout{No file #1.}}
}
\makeatother

\newcommand{\pos}{PoS System}
\newcommand{\psg}{\textsc{PSG}}
\newcommand{\proname}{\textsc{W2SB}}

\newcommand{\ourmeasure}{\textsc{C-NORM}}

\setcopyright{rightsretained}
\acmConference[GAIW'24]{Appears at the 6th Games, Agents, and Incentives Workshop (GAIW-24). Held as part of the Workshops at the 22st International Conference on Autonomous Agents and Multiagent Systems.}{May 2024}{Auckland, New Zealand}{Abramowitz, Aziz, Dickerson, Hosseini, Mattei, Obraztsova, Rabinovich, Tsang, Wąs (Chairs)} 
\acmSubmissionID{}

\title[A Game-Theoretic Analysis of Bootstrapping Protocols]{Centralization in Proof-of-Stake Blockchains: A Game-Theoretic Analysis of Bootstrapping Protocols}

\author{Varul Srivastava}
\affiliation{
  \institution{International Institute of Information Technology}
  \city{Hyderabad}
  \country{India}}
\email{varul.srivastava@research.iiit.ac.in}

\author{Sankarshan Damle}
\affiliation{
  \institution{International Institute of Information Technology}
  \city{Hyderabad}
  \country{India}}
\email{sankarshan.damle@research.iiit.ac.in}

\author{Sujit Gujar}
\affiliation{
  \institution{International Institute of Information Technology}
  \city{Hyderabad}
  \country{India}}
\email{sujit.gujar@iiit.ac.in}

\begin{abstract}
Proof-of-stake (PoS) has emerged as a natural alternative to the resource-intensive Proof-of-Work (PoW) blockchain, as was recently seen with the Ethereum Merge. PoS-based blockchains require an initial stake distribution among the participants. Typically, this initial stake distribution is called bootstrapping. This paper argues that existing bootstrapping protocols are prone to centralization. To address centralization due to bootstrapping, we propose a novel game $\Gamma_\textsf{bootstrap}$. Next, we define three conditions: (i) Individual Rationality (IR), (ii) Incentive Compatibility (IC), and (iii) $(\tau,\delta,\epsilon)-$ Decentralization that an \emph{ideal} bootstrapping protocol must satisfy. $(\tau,\delta,\epsilon)$ are certain parameters to quantify decentralization. Towards this, we propose a novel centralization metric, C-NORM, to measure centralization in a PoS System. We define a centralization game -- $\Gamma_\textsf{cent}$, to analyze the efficacy of centralization metrics. We show that C-NORM effectively captures centralization in the presence of strategic players capable of launching Sybil attacks. With C-NORM, we analyze popular bootstrapping protocols such as Airdrop and Proof-of-Burn (PoB) and prove that they do not satisfy IC and IR, respectively. Motivated by the Ethereum Merge, we study W2SB (a PoW-based bootstrapping protocol) and prove it is ideal. In addition, we conduct synthetic simulations to empirically validate that W2SB bootstrapped PoS is decentralized. 
\end{abstract}

\newcommand{\BibTeX}{\rm B\kern-.05em{\sc i\kern-.025em b}\kern-.08em\TeX}

\begin{document}

%%% The following commands remove the headers in your paper. For final 
%%% papers, these will be inserted during the pagination process.

\pagestyle{fancy}
\fancyhead{}

%%% The next command prints the information defined in the preamble.

\maketitle 

%%%%%%%%%%%%%%%%%%%%%%%%%%%%%%%%%%%%%%%%%%%%%%%%%%%%%%%%%%%%%%%%%%%%%%%%
%%%%%%%%%%%%%%%%%%%%%%%%%%%%%%%%%%%%%%%%%%%%%%%%%%%%%%%%%%%%%%%%%%=SECTION 1=%%%%%%%%%%%%%%%%%%%%%%%%%%%%%%%%%

%%%%%%%%%%%%%%%%%%%%%%%%%%%%%%%%%%%%%%%%%%%%%%%%%%%%%%%%%%%%%%%%%%%%%%%%
\section{Introduction}
\label{sec:introduction}
%%%%%%%%%%%%%%%%%%%%%%%%%%%%%%%%%%%%%%%%%%%%%%%%%%%%%%%%%%%%%%%%%%%%%%%%
\emph{Blockchain}, first introduced with Bitcoin \cite{Nakamoto2009}, is an append-only, distributed ledger. To maintain a distributed ledger, we need a \emph{distributed consensus} -- i.e., all the parties (which we call players) agreeing on a single state of the ledger. \emph{Public} blockchains employ consensus protocols such as Proof-of-Work (PoW), Proof-of-Stake (PoS), Proof-of-Elapsed-Time (PoET), among others~\cite{consensusSurvey}. All these protocols rely on \emph{incentive engineering} to achieve consensus. In PoW blockchains, the players solve cryptographic puzzles to propose the next block to be added to the ledger. In return, it obtains certain newly minted coins as a reward -- \emph{block reward}. With such {incentive engineering}, PoW achieves consensus if more than 50\% of players are honest, whereas traditional consensus protocols require $\frac{2}{3}^{\text{rd}}$ players to be honest~\cite{Castro1999pBFT}. The blockchain community is transitioning to more energy-efficient Proof-of-Stake (PoS) protocols (e.g. Ethereum~\cite{buterin2013ethereum}, the second largest blockchain by market capitalization and share~\cite{marketcap}, shifted to PoS via the Merge~\cite{ethereumMerge}). In PoS, a block proposer is selected with probability proportional to its stake in the system and awarded the new coins. With incentives involved, the strategic players would maximize their rewards, which might affect the security of the protocols. This paper focuses on centralization that might happen in PoS consensus protocols.

\noindent\textbf{Centralization in PoS blockchains -- Practical Examples}  Polygon Hard fork~\cite{governanceProblem} received criticism because it required signatures of just $13$ validators out of a quorum of $100$ due to centralization of resources. Furthermore, popular PoS-based blockchains such as ICON~\cite{icon}, Tezos~\cite{tezos}, Cosmos~\cite{cosmos}, and Irisnet are highly centralized (top $10$ stakeholders hold $\geq{1}/{3}$ of the total stake~\cite{centralizationSource}). This centralization exposes them to the risk of potential attacks\footnote{E.g., the Luna Crash~\cite{luna} resulted from offloading a massive volume of USDT in a very short interval. We suspect such coordinated de-pegging might be a repercussion of the centralization in the system.}.

\noindent\textbf{Why is Decentralization Important?} Governance and security of blockchains rely on the decentralization of resources among players. If a player/coalition collects disproportionate rewards, it might compromise the correctness of protocols~\cite{backbone,algorand}. In PoS blockchain, centralization is when a coalition (aka stake pool) holds a high amount of stake, exceeding their valuation for the system. A centralized PoS system is susceptible to threats such as denial of service (DoS), blacklisting, and double-spending~\cite{securitySurvey}, etc. 

\noindent Among the multiple causes of centralization in PoS, two major causes are (1) Centralization due to initial stake distribution --- called a \emph{bootstrapping} protocol, and (2) centralization due to institutions like exchanges and staking pools. This paper primarily focuses on bootstrapping. We argue that ineffective bootstrapping can cause centralization in PoS, implying the crucial role of bootstrapping protocol. Thus, there is a need to analyze decentralization via bootstrapping protocols in the presence of strategic players.

\subsection*{Our Approach} 

\paragraph{The Bootstrapping Model.} We begin by providing a novel construction of the bootstrapping of a PoS-based blockchain as a game, namely $\Gamma_\textsf{bootstrap}$ (Section~\ref{sec:game}) among participating players. We propose that the bootstrapping protocol should satisfy: (i) \emph{Individual Rationality} (IR) -- participating in the protocol should give higher utility than abstaining from it; (ii) \emph{Nash Incentive Compatibility} (IC) -- following the bootstrapping protocol should be a Nash Equilibrium for all players and, (iii) \emph{Decentralization} (DC), -- the PoS system should be decentralized after the bootstrapping protocol. We define a bootstrapping protocol to be \emph{ideal} if it simultaneously satisfies all three conditions (Definition~\ref{def:idealbootstrappingprotocol}). We need a measure to quantify centralization in \pos\ to assert decentralization. 

\paragraph{Quantifying Centralization.} An ideal bootstrapping protocol's central goal must be to distribute the initial stake such that the PoS system is not centralized. Towards this, we use the notion of $(\tau,\delta,\epsilon)-$decentralization~\cite{awesomeDecentralization}. Naturally, to quantify centralization, we require an effective metric. It must be resilient to a player's \emph{strategic} behaviour. Unfortunately, prior works on quantifying centralization in blockchains, such as using the Gini coefficient~\cite{GiniCoeff,gochhayatGini}, the Nakamoto Coefficient~\cite{centralizationMetricsGeneral2} or Entropy-based metrics~\cite{entropyDecentralization} are susceptible to manipulations by players who might game the metric to falsely make the system appear decentralized. 

\paragraph{\ourmeasure.} Given $\Gamma_\textsf{bootstrap}$, the definitions of IR and IC follow the standard template. However, it is challenging to construct a strategyproof centralization metric. To this end, we propose the construction of a $\mathcal{G}$ -- directed cyclic graph-- DAG-based representation of \pos\ transactions. $\mathcal{G}$ aims to represent (i) the current state and (ii) the history of the PoS system.
With $\mathcal{G}$, we propose our novel centralization metric, \ourmeasure\ (Section~\ref{ssec:ourmeasure}). \ourmeasure\ incorporates the strategic behaviour of the players. Moreover, the metric also captures the dynamic arrival of the players to the PoS system, which corresponds to dynamic updates to $\mathcal{G}$. 

\paragraph{The Centralization Game $(\Gamma_\textsf{cent})$.} 
We introduce $\Gamma_\textsf{cent}$, a game that provides a unified framework to measure the effectiveness of any centralization metric, given $\mathcal{G}$ (Section~\ref{ssec:decentralization-game}). $\Gamma_\textsf{cent}$ is a two-player game between a Metric Challenger and a Metric Descriptor. If the Metric Descriptor can distinguish between the centralized and decentralized (less centralized) systems presented by the challenger, then the descriptor wins the game. Otherwise, the metric is inept. We show that for famous blockchain centralization metrics -- Nakamoto coefficient~\cite{centralizationMetricsGeneral2}, Entropy~\cite{entropyDecentralization}, and Gini Coefficient~\cite{gochhayatGini}, the Metric Descriptor loses the game. However, for \ourmeasure, the Metric Descriptor wins with a very high probability (Theorem~\ref{thm:cnorm-game}). We also show that if a protocol is IC and has low \ourmeasure, then it satisfies our definition of $(\tau,\delta,\epsilon)-$ Decentralized (Theorem~\ref{thm:cnorm-decentralized}).

\noindent\textbf{Analyzing Bootstrapping Protocols.} Given $\Gamma_\textsf{bootstrap}, \mathcal{G}$ and $\Gamma_\textsf{cent}$, we proceed to discuss if the existing bootstrapping protocols are ideal or not (Section~\ref{sec:analysis-existing-protocols-w2sb}). First, we prove that \emph{Airdrop}~\cite{Airdrop} and \emph{Proof-of-Burn} (PoB)~\cite{PoB2020} are not ideal as they do not satisfy IC and IR, respectively. Second, we prove that a PoW-based bootstrapping protocol\footnote{The recent Ethereum Merge~\cite{ethereumMerge} motivates the study of PoW-based bootstrapping for PoS blockchains.} (abbreviated as \proname) simultaneously satisfies IR, IC, and DC and hence is ideal. Table~\ref{tab:comparison-table} summarizes these results.

We validate our findings in Section~\ref{sec:experimental-analysis} with experimental simulations of PoS blockchains bootstrapped with \proname\ for different valuation distributions and duration to achieve different levels of decentralization based on the random arrival of the players. % We observe that \ourmeasure\ value decreases on increasing the number of rounds, which backs the result of Theorem~\ref{thm:round-complexity}.

\paragraph{Our Contributions.} 
We make the following contributions.

\begin{enumerate}[leftmargin=*]
    \item We introduce a game-theoretic model for bootstrapping of \pos. Our model incorporates (i) dynamic participation and (ii)  Sybil-attacks by allowing coalition formation. 
    
    \item We define a bootstrapping protocol as ideal if it simultaneously satisfies IR and IC and $(\tau,\delta,\epsilon)$-DC (Definition~\ref{def:idealbootstrappingprotocol}).

    \item To verify if a protocol is decentralized, we introduce \ourmeasure\ (Definition~\ref{def:ourmeasure}). The definition relies on a Directed Acyclic Graph (DAG) based modelling of the PoS system.
  
    \item  We model the centralization game $\Gamma_\textsf{cent}$ to evaluate effectiveness of centralization metrics in Section~\ref{ssec:decentralization-game}. We show that when the metric descriptor chooses other metrics, they do not win the game while using \ourmeasure\ they win the game $\Gamma_\textsf{cent}$ with very high probability (Theorem~\ref{thm:cnorm-game}). We also show that a PoS System is $(\tau,\delta,\epsilon)-$decentralized if \ourmeasure\ value is low (Theorem~\ref{thm:cnorm-decentralized}). 
  
    \item Finally, we analyze existing bootstrapping protocols to show that Airdrop is not IC (Claim~\ref{claim:airdrop-no-ic}) and Proof-of-Burn is not IR (Claim~\ref{claim:pob-no-ir}). We then show that \proname\ is IR and IC (Lemma~\ref{lemma:pow-ic-ir}) and can achieve any \ourmeasure\ value in $(0,1]$ within a finite number of rounds (Theorem~\ref{thm:round-complexity}) Based on these results, we show that \proname\ is an ideal bootstrapping protocol (Theorem~\ref{thm:w2sb-ideal}). 

\end{enumerate}

\begin{table}[t]
\begin{small}
\centering
%%%% use adjustbox to stop overflow
\adjustbox{max width=\columnwidth}{%
\begin{tabular}{c | c c c | c}
\toprule
\textbf{Bootstrapping Protocol} & \textbf{IR} & \textbf{IC} & \textbf{DC} & \textbf{Ideal Bootstrapping}\\
\midrule
AirDrop~\cite{Airdrop} & \checkmark & \color{red}$\times$ & \color{red}$\times^{\star}$ & \color{red}$\times$ \\
Proof-of-Burn~\cite{PoB2020} &  \color{red}$\times$ &  \checkmark & \checkmark$^\dagger$ & \color{red}$\times$  \\
\textbf{\proname} &  \checkmark &  \checkmark &  \checkmark & \checkmark \\ 
\bottomrule
\multicolumn{5}{l}{$\dagger$: Depends on the parent cryptocurrency} \\
\multicolumn{5}{l}{$\star$: We conjecture that it is decentralized, but no existing metric can yet assert to it}
\end{tabular}}
\caption{\proname\ compared with existing bootstrapping protocols. An Ideal bootstrapping protocol satisfies IR, IC, and DC. }\label{tab:comparison-table}
\end{small}
\end{table}
%%%%%%%%%%%%%%%%%%%%%%%%%%%%%%%%%%%%%%%%%%%%%%%%%%%%%%%%=SECTION 2=%%%%%%%%%%%%%%%%%%%%%%%%%%%%%%%%%

\section{Related Work}
\label{sec:related-work}

\paragraph{Bootstrapping in Blockchains.} Proof-of-Stake (PoS) blockchain protocols employ various methods to bootstrap their networks, e.g., Initial Coin Offering (ICO)~\cite{challengesICO,ICO2020Momtaz}, Airdrop~\cite{Airdrop}, Proof-of-Burn~\cite{PoB2020}, and others \cite{Ethanos2021,OmniledgerBootstrap,VaultBootstrap} -- which focus on optimizing miners computational, communication, and memory costs. Airdrops are vulnerable to Sybil Attacks~\cite{AirdropSybil} and Privacy Leaks~\cite{AirdropPrivacy}. Smartdrop~\cite{smartdropArticle}, lacks a game-theoretic analysis of its security and effectiveness. Ethereum~\cite{buterin2013ethereum,wood2014ethereum}, originally a Proof-of-Work (PoW) cryptocurrency, transitioned to a PoS-based blockchain in 2022 during the``Merge event''~\cite{ethereumMerge}. Though it was not by design, Ethereum's PoW to PoS transition appears serve as a good bootstrapping protocol. 

% \sg{aren't most of the things below repeat from introduction?}
Game-theoretic analysis of blockchain protocols is gaining traction~\cite{blockchainPoA,elaineShiTFM,DavideAAMAS22,kiayiasNash,roughgarden2,roughgardenTFM1} but still lacks a formal game-theoretic framework for analyzing bootstrap protocols.% and a formal notion of an ``ideal'' bootstrapping protocol. 
The need for such a formal study is crucial as bootstrapping is one of the significant causes of centralization in PoS-based blockchains~\cite{DecentralizationTrailofBits,centralizationSurvey,centralizationSource}.% While PoS-based blockchains provide a respectable level of security and scalability, centralization is often a cause of concern, primarily due to known limitations against Sybil attacks~\cite{PoSSybil,sybilSurvey}.% \sg{it shd be crisp n clear why in realm of so many centralization metrics, we need another metric}

\paragraph{Centralization in Networks and Multi-Agent Systems.} 
There has been extensive research on quantifying centralization in different multi-agent systems~\cite{multiAgentCentralizationGame,centralityGeneral1}. While some of these centrality measures and analyses are for dynamic systems such as social networks, where agents do not behave strategically to game the system~\cite{centralityShapley,betweennessCentralization,narhariCentralization}, there are centrality measures that discuss rational parties who try to game the metric~\cite{coalitionGTCentrality,centralizationSecurity,centralitySecurity2}. 
Among these, Istrate \emph{et al.}~\cite{coalitionGTCentrality} discuss parties trying to game the system through coalitions and analyze network centrality.% as \emph{Coalitional Skill Games}. 
%However, the proposed metrics are intractable in general. 
%\sg{intractability need not be big concern if i give u enough computing power, will this metric solve the problem} 
However, the work quantifies centrality in networks, while we quantify in a system with each node having its own private and public valuations, forming coalitions among themselves (i.e., Sybil attacks). Moreover, it is intractable to compute the metric. 
Other metrics such as~\cite{centralizationSecurity} model the game as a Stackelberg Game where parties are trying to spoof the metric. However, the players are restricted to changing links in the networks, unlike our system, where they can report false valuations and perform stake redistribution (equivalent to weight redistribution among network nodes). Several metrics for measuring centralization exist in the blockchain literature as well \cite{GiniCoeff,gochhayatGini,centralizationMetricsGeneral1,centralizationMetricsGeneral2,entropyDecentralization}. E.g., metrics based on the Gini Coefficient~\cite{GiniCoeff} or the Nakamoto Coefficient~\cite{centralizationMetricsGeneral2}. We analyze the efficacy of these metrics in Section~\ref{sec:cnorm} and show that they fail in particular cases.

Kwon \emph{et al.}~\cite{awesomeDecentralization} focuses on achieving perfect decentralization in blockchain protocols such as PoW, PoS, and dPoS. %The authors define conditions for which a system is decentralized. They observe that (i) when Sybil cost is positive, it is possible to achieve perfect decentralization, whereas (ii) with no Sybil cost, it is almost impossible. 
We note two significant differences between  \citet{awesomeDecentralization} and our work. (i) we analyze centralization caused during bootstrapping whereas they analyze centralization in PoW and PoS systems post-bootstrapping. (ii) Our utility model penalizes centralization% comprises strategic players who may try to spoof the metric/protocol to fabricate reality and mimic a decentralized system 
% through a novel metric for quantifying centralization and capturing attempts for Sybil attack 
while Kwon \emph{et al.} assume the protocol is Sybil-resistant, which is a strong assumption.

Rauchs and Hileman~\cite{Hileman17} study centralization in the application layer due to the concentration of stake with exchanges and third-player wallet services. Further, there also exists analysis on the governance layer of different blockchain protocols~\cite{empericalGovernance,isBitcoinDecentralized}. Note that \ourmeasure\ is more general, encompassing the entire PoS system and capturing the stake distribution across different accounts. Other forms of centralization, e.g., due to increased storage costs \cite{storageCentralization,storageCentralizationMetric}, are out of the scope of this paper. For further details, an extensive survey on blockchain centrality measures can be found in~\cite{centralizationSurvey}.

%%%%%%%%%%%%%%%%%%%%%%%%%%%%%%%%%%%%%%%%%%%%%%%%%%%%%
\section{Preliminaries}
\label{sec:prelims}
%%%%%%%%%%%%%%%%%%%%%%%%%%%%%%%%%%%%%%%%%%%%%%%%%%%%%

% This section introduces (i) the notations used in this paper and (ii) the essential concepts related to blockchain technology and game theory through which we intend to establish understanding of concepts that will aid in the comprehension of our work. This section serves as a necessary precursor to the subsequent sections of the paper, where we delve into the specific details and contributions of our research. 

In this section, we (i) summarize relevant blockchain preliminaries and (ii) existing centralization metrics in the literature.

\subsection{Blockchain Preliminaries}\label{ssec:blockchain-prelims}

A blockchain system $\mathcal{B}$ is a decentralized ledger maintained by interested players $\mathcal{P} = \{p_{1},p_{2},\ldots,p_{n}\}$. They need to agree on the state of the ledger which can be achieved through Proof-of-Work (PoW), Proof-of-Stake (PoS), Proof-of-Burn (PoB) etc. We briefly discuss such consensus algorithms below.

\paragraph{Proof-of-Work~(PoW)~\cite{Nakamoto2009}.}
A blockchain-based on PoW comprises a consensus mechanism in which players (aka \emph{miners}) compete to solve a cryptographic puzzle to propose the next block. 
%The miners receive a block reward for each block they propose. In addition, they also incur a `mining' cost associated with the purchase and maintenance of mining equipment. 
A miner mines blocks in discrete periods called \emph{round}. If a miner's query is successful, it gets a reward $r_{\text{b}}$ in that round. Each miner $p_{i}$'s probability of mining a block is $a_{i}$, the fraction of the total mining power it controls. %Thus, each miner $p_{i}$'s expected reward is $a_{i}\cdot r_{\text{b}}$. 
Next, let $\chi$ denote the expected cost incurred per unit of mining power in a single round. If the total mining power is $M$, then the cost incurred in one round by miner $p_{i}$ is $a_{i}M\cdot\chi$. The expected payoff for the miner $i$ becomes $a_{i}\cdot(r_{\text{b}} - M\chi)$. % We call $(a_{i}\cdot r_{\text{b}}, a_{i}\cdot\chi)$ as the payoff vector for miner $i$ in a PoW-based blockchain.

\paragraph{Proof-of-Stake~(PoS)~\cite{posOuroboros2018}.}
The energy-intensive nature of PoW has led to PoS-based blockchains gaining traction due to their significantly lower carbon footprint~\cite{energyMerge,energyBlockchain}. 
In PoS, players `enroll' themselves for block creation by locking their \emph{stake}. Then, the probability of a player getting selected to propose the next block is proportional to its relative stake. %Players with a higher stake are incentivized to follow the protocol truthfully. 
PoS-based blockchains face several challenges, including centralization, fairness, and security concerns~\cite{centralizationSource}. Notice that a `decentralized' PoS-based blockchain is one where each player has a stake proportional to its \emph{valuation} for the system. This valuation, in theory, is reflected by the stake held (equivalently, the money invested) by the player. Thus, for any player $i$, the ratio of the stake $\omega_{i}$ held by it and its valuation $\theta_{i}$ should be the same for all $p_{i}$. We call this ratio as a \emph{true effective stake} of a system, and for each player, $i$ denote it as $\beta_{i} = \frac{\omega_{i}/\theta_{i}}{\sum_{j\in[n]} \omega_{j}/\theta_{j}}$. % In the rest of the paper, we refer to such a setting as $\mathcal{E} = \left(\mathcal{P},\{\omega_{1},\omega_{2},\ldots,\omega_{n}\},\{\beta_{1},\beta_{2},\ldots,\beta_{n}\}\right)$. \vs{to do: resolve this notational issue}
%$ = \{\mathcal{P},\}$.
%
PoS blockchains require initial stakes distribution for the consensus mechanism to start, achieved through bootstrapping. 
We now discuss some bootstrapping protocols employed in practice.

\paragraph{Proof-of-Burn~(PoB)~\cite{PoB2020}.} 
%This is a popular alternative to PoW and PoS and is more popularly used to bootstrap cryptocurrencies. 
In PoB, players `burn' a fraction of their cryptocurrency i.e. transfer of the currency to a provably unspendable address~\cite{PoB2020}.  
In PoB bootstrapped PoS system, a player burns $y$ fraction of a certain old cryptocurrency and obtains $x$ fraction of bootstrapped cryptocurrency's tokens.
%Players join a new (bootstrapped) cryptocurrency by burning tokens in some existing cryptocurrency. Consequently, we interest ourselves in PoB's use as a bootstrapping protocol. As an illustration, consider bootstrapping a cryptocurrency with native tokens -- $NTC$. The bootstrapping protocol can allocate stake worth $x\cdot NTC$ to players who burn $y\cdot OTC$ in some pre-decided existing cryptocurrency (with tokens represented as $OTC$) for some $x,y \in \mathbb{R}_{> 0}$. 

\paragraph{Airdrop~\cite{AirdropSybil}.} Airdrop distributes digital tokens/assets to eligible players. The eligibility is based on certain predetermined conditions specified by the cryptocurrency organization (e.g., posting a tweet). Airdrop, while useful for marketing a new cryptocurrency, suffers from privacy problems and is prone to \emph{Sybil attacks}~\cite {AirdropPrivacy,AirdropSybil}.

% Airdrop is a distribution mechanism widely used in cryptocurrency to distribute tokens or digital assets to a targeted audience. Airdrop involves distributing tokens to players eligible based on certain predetermined conditions specified by the cryptocurrency organization (e.g., posting a tweet). Although Airdrop is a great method for adopting and marketing a newly launched cryptocurrency, it suffers from privacy problems and is prone to \emph{Sybil attacks}~\cite {AirdropPrivacy,AirdropSybil}.

\paragraph{Sybil Attack~\cite{sybilAttack2002}.}
Sybil attack is when an attacker creates pseudo-identities that enable them to gain higher utility.% than using a single identity. 
%This utility is system-dependent, and increased utility might indicate higher reward, higher influence over the system, etc. 
In blockchain systems, a Sybil attack increases the reward the attacker obtains. A blockchain with player $p_{i}$ and pseudo-identities $e_{1},e_{2},\ldots,e_{k}$ and utility for $k$ being $U_{k}$ is prone to Sybil-attack if $U_{p_{i}} < \sum_{j=1}^{k}U_{e_{j}}$.

% Here, $U_{k}$ indicates utility obtained by posing identity $k$. 

% Since centralization poses a very real threat to the security of blockchains~\cite{DecentralizationTrailofBits}, there have been multiple metrics proposed to quantify the level of centralization in blockchain systems. Since the goal of an ideal bootstrapping protocol is to result in a decentralized protocol, quantifying centralization is an important aspect of our analysis. We next describe existing centralization metrics for the same.

\subsection{Centralization Metrics}\label{ssec:metrics-prelims}
% Blockchain protocols such as Proof-of-Work (PoW), Proof-of-Stake (PoS), Byzantine Fault Tolerance (BFT) based protocols, and Proof of Elapsed Time (PoET) involve making trade-offs among security, scalability, and decentralization~\cite{trilemmaSolver}. 
% \sd{why is this relevant?}
% \st{We commonly measure scalability in terms of Transactions Per Second (TPS). At the same time, security models such as the Universal Composable (UC) model and other existing methods}~\cite{BitcoinUC,UCCanetti,kiayiasNash} \st{are used to assess security. Similarly, different metrics exist to quantify decentralization in existing blockchains; details next.}

\noindent The following metrics are used to study centralization in blockchains. 
%\noindent Centralization poses a security threat to blockchains~\cite{DecentralizationTrailofBits}, different metrics quantify decentralization in blockchains; details next.

\paragraph{Entropy-based Centralization.}
It is an information-theoretic metric to quantify centralization~\cite{centralizationMetricsGeneral1,centralizationMetricsGeneral2,entropyDecentralization}. The entropy of a system $H$, normalized to the range $[0,1]$ is
\begin{equation}\label{eqn:entropy}
    H := - \frac{1}{\log(|P|)}\sum_{p_{i} \in P} \beta_{i} \log(\beta_{i})
\end{equation}

\paragraph{Gini Coefficient.}
It is a classic measure from economics that quantifies inequality in resource distribution using $G \in [0,1]$ where $0$ represents perfect equality, and $1$ represents maximal inequality. 
%The Gini coefficient is recently used for measuring centralization in blockchain. 
The Gini coefficient for a blockchain system $\mathcal{B}$ is~\cite{GiniCoeff,gochhayatGini,centralizationMetricsGeneral1,centralizationMetricsGeneral2}
\begin{equation}\label{eqn:gini}
    G := \frac{1}{2|P|} \sum_{p_{i} \in P}\sum_{p_{j} \in P} |\beta_{i} - \beta_{j}|
\end{equation}

\paragraph{Nakamoto Coefficient.}
Nakamoto coefficient~\cite{centralizationMetricsGeneral1,centralizationMetricsGeneral2,centralizationBlog}, a recent metric for measuring decentralization, is the minimum number of players in the system that together control the majority (i.e. more than some threshold $\tau_{th} \in [0,1]$) of the system. $\tau_{th}$ is protocol-dependent and is usually the minimum control required to disrupt the protocol. Mathematically,
%$p_{i}$ being $\beta_{i}$ as:
for a blockchain system $\mathcal{B}$, $N$ is,
\begin{equation}\label{eqn:nakamoto-coefficient}
    N := \min \left\{|A|\;:\;\sum_{p_{i} \in A} \beta_{i} > \tau_{th}, A \in 2^{P}\right\} 
\end{equation}

%%%%%%%%%%%%%%%%%%%%%%%%%%%%%%%%%%%%%%%%%%%%%%%%%%%%%
\section{Bootstrapping in \pos}
\label{sec:game}
%%%%%%%%%%%%%%%%%%%%%%%%%%%%%%%%%%%%%%%%%%%%%%%%%%%%%
In this section we formulate bootstrapping as a game $\Gamma_{bootstrap}$ to study the influence of rational agents (which can launch sybil attack) on the performance of the bootstrapping protocol. Following this, we discuss the game-theoretic properties desired from this Game and correspondingly define the requirements for a bootstrapping protocol to be an \emph{ideal bootstrapping protocol}.

%Let $P = \{p_{1},p_{2},\ldots p_{n}\}$ be the set of players interested in being part of the Proof-of-Stake (PoS) based, decentralized blockchain system  -- \pos. 
Let $\Pi$ be a bootstrapping protocol to mint the initial coins in \pos, which runs till round $T$, which we call \emph{stopping time}. After $T$ rounds, $\Pi$ terminates, and the system runs a PoS protocol. %Note that $\Pi$ here is some bootstrapping protocol. 
Each player $p_{i}\in P$ has a private valuation of \pos, $\theta_{i} \in \mathbb{R}_{\geq 0}$, on joining the system. $\theta_i$ indicates the amount that player $p_i$ will prefer to invest in \pos. The private valuation profile of all the agents is $ \bm{\theta} := (\theta_{1},\theta_{2},\ldots, \theta_{n})$. 

% \sd{not clear}
Each player $p_{i}$ elicits its valuation as $\hat{\theta}_{i} \in \mathbb{R}_{\geq 0}$ (through some preference aggregation procedure which is part of $\Pi$). E.g., by simply obtaining the coins of the underlying cryptocurrency. The set of reported valuations is $\bm{\hat{\theta}} := ( \hat{\theta}_{1}, \hat{\theta}_{2}, \ldots \hat{\theta}_{n} )$. In a \pos, the goal is to ensure each player holds a stake proportional to its valuation of the system (see discussion in Section~\ref{ssec:blockchain-prelims}). Thus, $\Pi$ must ensure that at the end of round $T$, the stake held by each player $p_{i}$ is proportional to their reported valuation of the system, i.e., $\hat{\theta}_{i}$. 

Moreover, a subset of players may collude to earn coins more than their fair share. We assume that the valuations are \emph{additive}, so the total true (private) valuation of any collusion set $A\subset P$ is $\vartheta(A) = \sum_{i \in A} \theta_{i}$. To model such strategic behavior, our model incorporates colluding as part of the player's strategy space. Certainly, collusion may lead to centralization in the \pos. We assume that $\Omega$ is a measure available with the \pos\ to quantify centralization. The \pos\ $\mathcal{B}$ is represented by the tuple $(\Pi,\Omega,P)$.

% \sd{repeated in the immediate line as well}
% In summary, each player $p_{i}$ in the game can strategize (i) on reporting $\hat{\theta}_i$ and (ii) deciding which coalition to join.
%%%%%%%%%%%%%%%%%%%%%%%%%%%%%%%%
\subsection{Strategy Space}
\label{ssec:strategy}
%%%%%%%%%%%%%%%%%%%%%%%%%%%%%%%%
%The game is defined given a bootstrapping protocol $\Pi$ and centralization metric $\Omega$. 
As players are strategic, each player $p_i$'s goal in the system $\mathcal{B}$ is to maximize its stake while keeping the recorded value of the system's centralization (or $\Omega$) as low as possible. Towards this, it makes two choices, (1)  $\hat{\theta}_{i}$ and (2) if it should collude with the others. 

\paragraph{Sybil-attack.} The set $P$ denotes all the identities part of the \pos. However, if $p_{i}$ has launched a Sybil attack and created pseudo-identities, we represent them in our game as partition sets --- each partition being the group of sybil identities (sybil-group) of one rational player. These sybil-groups are labelled and represented by the set $\mathcal{A} := \{A_{1}, A_{2},\ldots, A_{n}\}$ where $A_{i}$ is the partition (name of the set) which $p_{i}$ is part of. $\mathcal{A}$ follows following three properties: (1) $|A_{i}| \in \{1,\ldots,n\}$. (2) $A_{i} \cap A_{j} = \emptyset$ if $p_i$ and $p_j$ are in different partition and $A_{i} = A_{j}$ otherwise. (3) $\cup_{i=1}^{z} A_{i} = P$. Through such collusion, players can perform redistribution of stake, making the system appear more decentralized to the centralization metric $\Omega$. The redistribution of the stake cannot increase the allocated stake; therefore, $\sum_{i \in A} \hat{\theta}_{i}$ remains constant. To summarize, the strategy space for $p_i$ is $M = \mathbb{R}_{\geq 0}\times\;\mathcal{A}_{i}$; a cartesian product of reported stake and set of all possible partitions containing $p_{i}$.

%%%%%%%%%%%%%%%%%%%%%%%%%%%%%%%%
\subsection{Utility Structure}
\label{ssec:utility-structure}
%%%%%%%%%%%%%%%%%%%%%%%%%%%%%%%%

Given a \pos\ system $\mathcal{B}$, we now quantify players' utilities. Each player's utility depends on (1) $\Pi$ and (2) the level of system centralization (measured by $\Omega$), as the exchange rate for the cryptocurrency and, therefore, the utility drops for a more centralized (and therefore prone to attacks and censorship) cryptocurrency. %Each strategic player aims to gain a maximal stake while making the protocol appear decentralized. All players prefer to keep $\Omega$ as low as possible.

The earned reward for any player $p_i$ depends on its reported valuation and the net gain of the PoS system. This reward is the difference between the total collective rewards made by the \pos\ ($r_{\text{block}}$) and the total cost incurred by it ($b_{\text{spent}}$). We capture the loss in the stake's external value due to centralization via a function $g(\theta_{i})$ that is non-decreasing in $\theta_{i}$. Thus, for the described setup, the utility for each player $p_{i}\in P$ with strategy ($\hat{\theta}_{i}, A_{i}$) is,
\begin{equation}\label{eqn:utility-general}
U_{i}(\hat{\theta}_{i},\hat{\bm{\theta}}_{-i},A_{i},\mathcal{A}_{-i};\theta_{i} ) := b\cdot\hat{\theta}_{i} - \Omega(\cdot)\cdot g(\theta_{i})
\end{equation}

In Eq.~\ref{eqn:utility-general}, $\mathcal{A}_{-i} = \mathcal{A}/\{A_{i}\}$ and $\hat{\bm{\theta}}_{-i} = \hat{\bm{\theta}}/\{\hat{\theta}_{i}\}$. The reported valuation (aggregated through $\Pi$) is $\hat{\bm{\theta}}$ and the true valuation is $\bm{\theta}$. The value $b$ depends on $\Pi$ and is independent of player strategy. $b = (r_{\text{b}} - b_{\text{spent}})\cdot\gamma$ where $r_{\text{b}}, b_{\text{spent}}$ and $\gamma$ are some constants dependent on $\Pi$. The total reward for player $p_{i}$ is proportional to the stake allocated to it (which in turn is proportional to $\hat{\theta}_{i}$). The \emph{cost of centralization} depends on $g(\theta_{i})$, which is a non-decreasing function of $\theta_{i}$ for player $p_{i}$. Note that $\Pi$ and $\Omega$ would be typically fixed for a blockchain system under study, so whenever clear from context, we omit mentioning these two in $U_i(\cdot)$.

As forming a collusion is part of strategy space, we can capture Sybil attacks, i.e., a player creating multiple identities and forming collusion with these identities.

In summary, the bootstrapping model for PoS system $\mathcal{B} = (\Pi,\Omega,P)$ can be defined by $\Gamma_\textsf{bootstrap} = \langle P, \bm{\theta},(U_{i})_{i \in [n]}\rangle$.

\subsection{Ideal Bootstrapping Protocol}
\label{ssec:ideal-general-def}
%%%%%%%%%%%%%%%%%%%%%%%%%%%%%%%%
Given $\Gamma_{\textsf{bootstrap}}$, we now enlist the properties that $\Pi$ should satisfy to be an \emph{ideal} bootstrapping protocol. 

\subsubsection{Individual Rationality (IR).}  -- Each player $p_{i}$ prefers participating in $\Pi$ over abstaining from participation. The strategy of abstaining for any player $p_{i}\in P$ is $\hat{\theta}_{i} = 0$ and $A_{i} = \{i\}$. Formally, 
\begin{definition}[Individual Rationality (IR)]\label{def:ir}
$\Pi$ is Individually Rational (IR) if %honest participation in the protocol gives utility greater than when abstaining from participation.
%(i.e., following $\hat{\theta}_{i} = 0, A_{i} = \emptyset$) for any given player type, i.e., for any $\hat{\theta_i} \in \mathbb{R}_{\geq 0}$ and any strategy followed by the remaining players. Mathematically, 
$\forall p_{i} \in P,\;\forall \bm{\theta} \in \mathbb{R}^{n}_{\geq 0},\;\forall\hat{\bm{\theta}}_{-i} \in \mathbb{R}^{n-1}_{\geq 0},\;\forall \mathcal{A}_{-i}$, 
\begin{equation}\label{eqn:ir}
    U_{i}(\theta_{i},\hat{\bm{\theta}}_{-i},\{i\},\mathcal{A}_{-i};\theta_{i}) > U_{i}(0,\hat{\bm{\theta}}_{-i},\{i\},\mathcal{A}_{-i};\theta_{i})
\end{equation}
\end{definition}

\subsubsection{Incentive Compatibility (IC)} 
%We also require $\Pi$ to be Incentive Compatible. The equilibrium strategy for each player $p_{i}$ should be to report their true valuation, i.e.
IC requires, for each player $p_i$,$\hat{\theta}_{i} = \theta_{i}$ should give maximum expected utility given all other players are reporting their true valuations.%, given $\Pi$ and $\Omega$. This condition is defined as:
\begin{definition}[Nash Incentive Compatibility (IC)]\label{def:ic}
$\Pi$ is nash incentive compatible if $\forall p_{i} \in P, \forall \bm{\theta} \in \mathbb{R}^{n}_{\geq 0}, \forall \hat{\theta}_{i} \in \mathbb{R}_{\geq 0}, \forall A_{i} \in \mathcal{A}_{i}$, where $\mathcal{H} := \{\{1\},\{2\},\ldots,\{n\}\}$
\begin{equation}\label{eqn:ic}
    \mathbb{E}[U_{i}(\theta_{i},\bm{\theta}_{-i},\{i\},\mathcal{H}_{-i};\theta_i)] \geq \mathbb{E}[U_{i}(\hat{\theta}_{i}, \bm{\theta}_{-i},A_{i},\mathcal{{H}}_{-i}
    ;\theta_i)] 
\end{equation}
\end{definition}
%\sg{utility is $\mid$ Pi n Omega, but here those two terms dropped...one possible thing in the para below eqn 5, mention that as we  Pi n  Omega are fixed for analysis, we drop them from the notation}
\noindent The expectation taken is over the randomness in the system due to (i) random dynamic arrival of players and (ii) randomness of $\Pi$.

\subsubsection{Decentralization.} 
Lastly, we desire that the \pos\ be decentralized. For this, we motivate our definition partially from~\cite[Definition 4.1]{awesomeDecentralization} along with the condition of Sybil-resistance. We require (1) at least a fraction $\tau \in (0,1]$ of the total players to have joined the system, (2) the \emph{true scaled stake}\footnote{In PoS blockchain $\mathcal{B}$, $\beta_{\max}:=\max_{i\in[n]} \beta_{i}$ and $\beta_\delta:= \beta_{r}$ where $p_{r}$ is $\delta^{th}$ percentile player when arranged in decreasing order of stake.%. i.e. $\beta_{max} = \beta_{\delta=100}$ and $\beta_{min} = \beta_{\delta=0}$.
} for each player $p_{i}$ should be as close to each other as possible (3) the system should be Sybil-resistant.%, i.e., players should not be able to reduce the difference in the true scaled stake (and thereby reducing the centralization) by forging identities. 
Our definition differs from~\cite[Definition 4.1]{awesomeDecentralization} as we capture each player's valuation using \emph{scaled stake} instead of the total stake held by the player. We also impose a condition of Sybil-proofness. We define decentralization as follows:

% \sd{i could not follow the definition of $\beta_{\max}$ and $\beta_{\delta}$.. where are they introduced?}

\begin{definition}[$(\tau,\delta,\epsilon)-$Decentralization]\label{def:decentralization}
    A PoS based blockchain system is said to be $(\tau,\delta,\epsilon)-$ \emph{decentralized} for $\tau \in (0,1]$, $\delta \text{ (percentile)} \in [0,100]$ and $\epsilon \in \mathbb{R}_{\geq 0}$ if it follows: 
    \begin{itemize}
        \item \textbf{Minimum Participation:} %The fraction of total players from set $P$ who have joined the system by 
        The current set of players joined the PoS system by round $t$, $P_{t}$, is such that $\frac{|P_{t}|}{|P|} \geq \tau$. 
        \item \textbf{Proportionality:} The ratio of max true scaled stake ($\beta_{max}$) and the $\delta-$percentile true scaled stake ($\beta_{\delta}$) is $\frac{\beta_{max}}{\beta_{\delta}} \leq 1 + \epsilon$.
        \item \textbf{Sybil-proofness:} The ratio $\frac{\beta_{max}}{\beta_{\delta}}$ cannot be reduced by a single player forging identities without decreasing its utility. 
    \end{itemize}
\end{definition}

\smallskip \noindent Now we define an \emph{Ideal bootstrapping} protocol for a \pos. 

\begin{definition}[Ideal Bootstrapping Protocol]\label{def:idealbootstrappingprotocol}
    Given a set of players $P$ participating in a \pos\ with bootstrapping protocol $\Pi$, we call $\Pi$ as an \emph{ideal bootstrapping protocol} under a proposed metric $\Omega$ of measuring centralization of the PoS system, if it satisfies: 
    \begin{itemize}
        \item $\Pi$ is \emph{individually rational} for each player $p_{i} \in P$.
        \item $\Pi$ is Incentive Compatible, i.e. $\hat{\bm{\theta}} = \bm{\theta}$ is equilibrium strategy.
        \item After executing $\Pi$, the \pos\ is $(\tau,\delta,\epsilon)-$Decentralized.
    \end{itemize}
\end{definition}

\noindent To measure how decentralized a PoS system is, we need a metric to quantify centralization, which is the focus of the next section.% In the next section, we propose \ourmeasure, a novel metric to quantify centralization in a PoS System. 
% We also show why \ourmeasure\ is more effective than other metrics used to quantify blockchain centralization. 

%%%%%%%%%%%%%%%%%%%%%%%%%%%%%%%%%%%%%%%%%%%%%%%%%%%%%%%%%%%%%%%%%
\section{Measuring Centralization in a \pos}
\label{sec:cnorm}
%%%%%%%%%%%%%%%%%%%%%%%%%%%%%%%%%%%%%%%%%%%%%%%%%%%%%%%%%%%%%%%%%
In this section we (i)model a given \pos\ $\mathcal{B}$ as a \emph{PoS System Graph} (\psg), (ii) discuss \ourmeasure\ and its properties..

%%%%%%%%%%%%%%%%%%%%%%%%%%%%%%%%%%%%%%%%
\subsection{PoS System Graph}
\label{ssec:pos-graph}
%%%%%%%%%%%%%%%%%%%%%%%%%%%%%%%%%%%%%%%%

Given a complete history of transactions in a PoS system $(\mathcal{B})$ (which is available as a public ledger), we discuss how we construct a \psg\ $(\mathcal{G})$ that captures strategic redistribution of stake among the colluding players. We later use $\mathcal{G}$ to define our metric \ourmeasure.% for measuring the centralization of $\mathcal{B}$.

\subsubsection{Motivation.} The purpose of the \psg\ $(\mathcal{G})$ is to capture the following two properties of the \pos: (1) the current state stake owned by a particular player -- captured through the weights of each vertex in $\mathcal{G}$ and (2) the history, i.e., redistribution of stake -- captured through the weighted directed edges between different vertexes. Additionally, $\mathcal{G}$'s construction must only depend on publicly available information about $\mathcal{B}$.

\subsubsection{PoS System Graph (\psg).}  The Graph $\mathcal{G} = ([n], C, W)$ is a tuple consisting of a set of vertices $[n]$ where each vertex corresponds to a player in the system. Each vertex (player) $p_{i}$ has a total stake $c_{i}$ with it. We assign weights to these nodes as $\bm{C} =(c_1,c_2,\ldots,c_n)$. To capture the transactions summary between players $p_i$ and $p_j$, we place a weighted directed edge from vertex $i$ to $j$ with weight $w_{i,j}$. The set $W$ contains $w_{i,j}$ which are determined as follows:

\begin{itemize}[itemsep=0.3em,leftmargin=*]
    \item[$\bullet$] Define $\sum t_{i,j}$ as net positive transaction from $p_{i}$ (currently having stake $c_{i}$) to $p_{j}$ (with stake $c_{j}$) and $\sum t_{j,i}$ as net positive transaction from $p_{j}$ to $p_{i}$. If $\sum t_{i,j} > \sum t_{j,i}$ then there exists an edge from $p_{i}$ to $p_{j}$ with weight $w_{i,j} = \sum t_{i,j} - \sum t_{j,i}$ (see Figure~\ref{fig:graph-figure} in Appendix~\ref{app:cycle-elimination}). 
    \item[$\bullet$] If $\sum t_{i,j} = \sum t_{j,i}$ then there is no connection between $p_{i}$ and $p_{j}$ i.e. $w_{i,j} = w_{j,i} = 0$. By extension, $w_{i,i} = 0$ for all $i \in [n]$.
    \item[$\bullet$] Constructing such a graph might lead to a cycle (e.g. when $p_{1}$ transfers to $p_{2}$, $p_{2}$ transfers to $p_{3}$ and $p_{3}$ transfers to $p_{1}$). In such cases, a cycle-elimination procedure converts the graph to a Directed Acyclic Graph (DAG). % We eliminate cycles in any arbitrary order. 

    % \el{\paragraph{Note.} The order of elimination does not affect the value of our proposed centralization metric because it does not alter the $\sum t_{i,j} - \sum t_{j, i}$ values for any player $p_{i}$. Cycle Elimination is depicted through Figure~\ref{fig:graph2} (before eliminating cycles) and Figure~\ref{fig:graph3} (after cycle elimination). The \texttt{CycleElimination} Algorithm is provided in Appendix~\ref{app:cycle-elimination} for completeness.} % have discussed once cnorm is introduced.
\end{itemize}

%%%%%%%%%%%%%%%%%%%%%%%%%%POS-Centralization Graph Construction examples %%%%%%%%%%%%%%%%%%%%%%%%%%%%%%%
%% DO WE NEED THIS FIGURE 

% \iffalse
% \fi

%%%%%%%%%%%%%%%%%%%%%%%%%%%%%%%%
\subsection{\ourmeasure: Quantifying Centralization in a PoS Blockchain System}
\label{ssec:ourmeasure}
%%%%%%%%%%%%%%%%%%%%%%%%%%%%%%%%

We remark that most existing centralization metrics fail to capture either the collusion of players or the valuation of a player while calculating the centralization metric. 
% \st{Following the philosophy of PoS blockchains,} 
A system is decentralized if\footnote{this notion of decentralization follows from the idea of proportional inceitive towards security (incentive to maintain security in PoS for any player is proportional to their valuation of the system). It has not been proven or discussed if this notion of valuation-proportional stake is necessary, but it is sufficient.} the stake held by each player is proportional to their valuation of the system. 
We propose a new metric for measuring centralization in a \pos. Our metric uses the PoS system graph $\mathcal{G} = ([n],C,W)$, a \psg\ derived from $\mathcal{B}$, as described in Section~\ref{ssec:pos-graph}. 

Towards quantifying centralization, we first calculate the \emph{effective stake} of a player -- capturing the net stake and the transaction value between the player and its neighbours in $\mathcal{G} $. Considering $nbr(i)$ to be the set of vertices adjacent to vertex $i$, the effective stake $\omega_{i}$ of player $p_{i}$ is calculated as: % \sd{check below eqn for inconsistent notations}
% \begin{equation}\label{eqn:effective-stake}
$\omega_{i} := c_{i} + \sum_{j \in nbr(i)} w_{j,i} - \sum_{j \in nbr(i)} w_{i,j}$.
% \end{equation}

\noindent\emph{Effective stake per unit of valuation} is denoted by $\frac{\omega_{i}}{\theta_{i}}$. For each $p_i\in P$, the normalized effective stake is called \emph{scaled stake}, denoted by $\beta_i$. \begin{equation}\label{eqn:scaled-stake}
    \beta_{i}(\theta) := \frac{\omega_{i}/\theta_{i}}{\sum_{j \in [n]} \big(\omega_{j}/\theta_{j}\big)}
\end{equation}

Since the allocated stake should be proportional to the reported valuation of each player $p_{i}$, $\frac{\omega_{i}}{\theta_{i}}$ should be the same for all players. For an IC mechanism $\Pi$, we should ideally have $\frac{\omega_{i}}{\theta_{i}} = \frac{\omega_{j}}{\theta_{j}} \;\forall\;i,j\in[n]$. Thus, the scaled stake should have the same value for each player in the PoS system for the system to be fully decentralized. Upon calculation, we find that $\beta_{i}$ should (in the ideal case) be $\frac{1}{n}$. Thus, for any player $p_{i}$, the ``error'' is the deviation of their {scaled stake} from this target value (in a decentralized system) of $\frac{1}{n}$. Hence, we define our metric for measuring centralization as the summation of $L_{1}$-Norm of this deviation for each player. The normalization factor $\frac{1}{2}$ maps the range of $\Omega$ to $[0,1]$. Formally,

\begin{definition}[C-Norm $(\Omega)$]\label{def:ourmeasure}
    For a PoS system with player set $P$, PoS system graph $G$ the C-Norm of the system $\Omega : [n]\times C \times W \rightarrow [0,1]$ is defined over $\bm{\theta}$ as
$        \Omega(\bm{\theta}) := \frac{1}{2}\sum_{j=1}^{[n]}\Big|\beta_{j}(\bm{\theta}) - \frac{1}{n} \Big| $.
    The centralization of the system in the worst-case is $\Omega^{\star}$ as:
    % \sg{dont use the same notation $\Omega$ here...i had suggested probably use $\Omega^{\star}$}
    \begin{equation}\label{eqn:cnorm}
        \Omega^{\star} := \max_{\bm{\theta} | \theta_{i} > 0 \forall i \in [n]}\frac{1}{2}\sum_{j=1}^{[n]}\Big|\beta_{j}(\theta) - \frac{1}{n} \Big|
    \end{equation}
\end{definition}

\paragraph{Note.} While constructing the \psg\  $\mathcal{G}$, the order of elimination of cycles using the \texttt{CycleElimination} Algorithm (see Appendix~\ref{app:cycle-elimination}) has no effect on the value of \ourmeasure. This is because the difference between $\sum_{j \in Nbr(i)} t_{i,j} - \sum_{j \in Nbr(i)} t_{j,i}$ remains the same and therefore $\omega_{i}$, and by extension, $\Omega^{\star}$ does not change.

% We use this metric of centralization in our model of analysis of bootstrapping in PoS blockchains. Under \ourmeasure, we discuss necessary and sufficient conditions for the protocol to be decentralized. 

% As a corollary of Lemma~\ref{lemma:ic-1c-decentralized} we can rewrite the definition of \emph{Ideal Bootstrapping Protocol} (using Definition~\ref{def:idealbootstrappingprotocol}) for the proposed \ourmeasure\ metric of centralization as follows: 

% % \sd{lemmas together will lead to a theorem?} \vs{ Yes, W2SB is Ideal will follow from these lemmas.}

% \begin{corollary} 
%     Given a set of players $P$ participating in a PoS-based blockchain protocol, which protocol $\ Pi$ bootstraps, we call $\Pi$ as the \emph{ideal bootstrapping protocol} under \ourmeasure\ $\Omega_{\rho}$ (metric for centralization of the PoS system), if it satisfies the following: 
%     \begin{itemize}
%         \item The protocol $\Pi$ is \emph{ex-post individually rational} for each player $p_{i} \in P$.
%         \item The protocol $\Pi$ is Incentive Compatible, i.e. $\hat{\bm{\theta}} = \bm{\theta}$ is equilibrium strategy.
%         \item The system $(\Pi,\Omega_{\rho})$ is $1-$Colluding.
%     \end{itemize}
% \end{corollary}
% \vs{Added the below paragraph.}
\paragraph{Need for C-Norm:} \ourmeasure\ has few advantages over previous metrics. While other metrics evaluate centralization based on the stake of each player, \ourmeasure\ include player's \emph{valuation} of the system, therefore defining effective stake. In addition, while other metrics use the current state of the system, \ourmeasure\ uses previous state changes (transactions) as well, which is captured through \pos\ Graph. Next, we formally argue why the proposed metric is good at capturing centralization in a \pos.

\subsection{Evaluating Decentralization Metrics}
\label{ssec:decentralization-game}

We begin by introducing a novel \emph{centralization game}, which assesses the efficacy of a given centralization metric for a \pos. We illustrate with examples that widely used metrics (discussed in Section~\ref{sec:prelims}) namely Gini coefficient~\cite{GiniCoeff,gochhayatGini} ($G_{c}$), Nakamoto Coefficient~\cite{decentralizationMeasure,centralizationBlog} ($N$) and Entropy-based metrics~\cite{centralizationMetricsGeneral1,entropyDecentralization} ($H$) fail to capture centralization in the proposed game. A strategic player can create instances where metrics fail to detect centralization, whereas \ourmeasure\ detects them. 
Next, we also show that any \pos\  bootstrapped using an IC $\Pi$ and a small \ourmeasure\ satisfies $(\tau,\delta,\epsilon)-$decentralization (Definition~\ref{def:decentralization}).

%~~~~~~~~~~~~~~~~~~~~~~~~~~~~~~~~~~~~~~~~~~~~~~~~~~~~
%%%%%% Centralization Game %%%%%%%%%%%%
%%%%%
\begin{figure}[t]
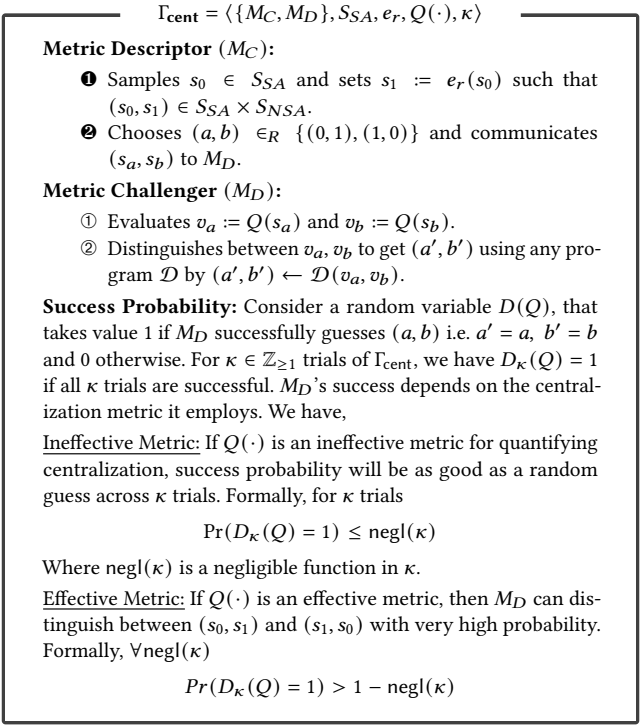

    \centering
    \begin{small}    
    \begin{mybox}{$\Gamma_{\textsf{cent}}=\langle \{M_{C},M_{D}\},S_{SA},e_{r},Q(\cdot),\kappa\rangle$}

    \textbf{Metric Descriptor $(M_{C})$:}
    \begin{itemize}
        \item[\ding{182}] Samples $s_{0} \in S_{SA}$ and sets $s_{1} := e_{r}(s_{0})$ such that $(s_{0},s_{1}) \in S_{SA}\times S_{NSA}$.
        \item[\ding{183}] Chooses $(a,b) \in_{R} \{(0,1),(1,0)\}$ and communicates $(s_{a},s_{b})$ to $M_{D}$.
    \end{itemize}

    \textbf{Metric Challenger $(M_{D})$:}
    \begin{itemize}
        \item[\ding{192}] Evaluates $v_{a} := Q(s_{a})$ and $v_{b} := Q(s_{b})$. 
        \item[\ding{193}] Distinguishes between $v_{a},v_{b}$ to get $(a',b')$ using any program $\mathcal{D}$ by $(a',b')\gets\mathcal{D}(v_{a},v_{b})$.
    \end{itemize}
    \textbf{Success Probability:} Consider a random variable $D(Q)$, that takes value $1$ if $M_D$ successfully guesses $(a,b)$ i.e. $a' = a,\;b' = b$ and $0$ otherwise. For $\kappa\in\mathbb{Z}_{\geq 1}$ trials of $\Gamma_\textsf{cent}$, we have $D_\kappa(Q)=1$ if all $\kappa$ trials are successful. $M_D$'s success depends on the centralization metric it employs. We have,
    
    \smallskip
    \underline{Ineffective Metric:} If $Q(\cdot)$ is an ineffective metric for quantifying centralization, success probability will be as good as a random guess across $\kappa$ trials. Formally, for $\kappa$ trials
    \[\Pr(D_\kappa(Q) = 1) \leq \textsf{negl}(\kappa) \]
    Where $\textsf{negl}(\kappa)$ is a negligible function in $\kappa$.
    
    \smallskip
    \underline{Effective Metric:} If $Q(\cdot)$ is an effective metric, then $M_{D}$ can distinguish between $(s_{0},s_{1})$ and $(s_{1},s_{0})$ with very high probability. Formally, $\forall \textsf{negl}(\kappa)$
    \[Pr(D_\kappa(Q) = 1) > 1 - \textsf{negl}(\kappa)\]
    \end{mybox}
    \caption{The Centralization Game}
    \label{fig:cent-game}
    \end{small}
\end{figure}
%~~~~~~~~~~~~~~~~~~~~~~~~~~~~~~~~~~~~~~~~~~~~~~~~~~~~
%%%%%% Centralization Game %%%%%%%%%%%%
%%%%%
%%

\subsubsection{Centralization Game ($\Gamma_{cent}$).} The game has two players -- Metric Descriptor $M_{D}$ and Metric Challenger $M_{C}$. There are two types of \pos. First, $S_{SA}$, which are centralized systems, with players (part of the \pos) launching Sybil attacks to make them appear decentralized. It means $S_{SA}$ is a set of \psg s which are DAGs with at least one directed edge (having non-negligible weight). Second, $S_{NSA}$ are decentralized systems without Sybil attack. Additionally, there exists an onto function $e_{r}:S_{SA}\longrightarrow S_{NSA}$ which removes edges from the \psg\ $\mathcal{G}$ for the given $s_{0} \in S_{SA}$. 

The Game is repeated for $\kappa\in\mathbb{Z}_{\geq 1}$ trials. At any trial, $M_{C}$ gives randomly shuffled samplings $s_{0} \in S_{SA}$ and $s_{1} = e_{r}(s_{0})$. $M_{D}$ uses its centralization metric under analysis $Q(\cdot)$ to find the centralization value for both $s_0$ and $s_1$ and uses these values as inputs to any program $\mathcal{D}$ to report the ordering which was chosen by $M_{C}$. $M_D$ is successful if it correctly guesses the ordering for all $\kappa$ trials. We use $\Gamma_{cent}:=\langle \{M_{C}, M_{D}\}, S_{SA}, S_{NSA}, e_{r}, Q(\cdot), \kappa\rangle$ to refer to this centralization game. Figure~\ref{fig:cent-game} formally defines $\Gamma_{cent}$.

% \begin{itemize}
%     \item[\ding{182}] $M_{C}$ samples a PoS system state $s_{0} \in S_{SA}$ and calculate $s_{1} = e_{r}(s_{0})$. Thus $M_{C}$ has $(s_{0},s_{1}) \in S_{SA}\times S_{NSA}$. $M_{C}$ chooses a random bits $(a,b) \in_{R} \{(0,1),(1,0)\}$ and send ($s_{a},s_{b}$) to $M_{D}$.
%     \item[\ding{183}] $M_{D}$ uses the proposed metric for quantifying centralization $Q(\cdot)$ (which is under analysis) to tell if $s_{a} \in S_{SA}$ or $s_{a} \in S_{NSA}$ i.e. $a = 0$ or $1$ (and correspondingly $b = 1$ or $0$). If $M_{D}$ reports correct $b$ we call it a \emph{success event} which corresponds to the random variable $D(Q) = 1$.
% \end{itemize}
% \noindent We observe that any metric where $M_{D}$ cannot distinguish between the bits with better \st{chance} than \sd{a} random guess \st{is a bad metric} \sd{may not be a useful metric} for quantifying centralization as it is unable to distinguish cases where \sd{the} adversary is trying to game the system from a truly decentralized system \sd{what do you mean}. Therefore, a \sd{centralization} metric $Q(\cdot)$ is \emph{not a good metric} for quantifying centralization \st{of} \sd{in a?} PoS system if $|Pr(D(Q)=1) - \frac{1}{2}|$ is negligibly small \sd{negliegible in the number of times the adversary plays $\Gamma_{cent}$} for any sampling $(s_{0},s_{1})$ by $M_{C}$. 

\paragraph{Comparing the Success of Centralization Metrics in $\Gamma_{\textsf{cent}}$.} We now analyze the existing metrics for quantifying centralization (refer Section~\ref{ssec:metrics-prelims}) along with \ourmeasure. Towards this, we first discuss the sampling (say) $(s_{0},s_{1})$ by $M_{C}$ and then calculate the values $Q(s_{0})$ and $Q(s_{1})$ for different metrics in Table~\ref{tab:metric-comparison}. We then argue why any program $\mathcal{D}$ cannot distinguish between $s_{0}$ and $s_{1}$ for other metrics except for \ourmeasure.

\begin{Example}[Analysing centralization metrics using $\Gamma_{cent}$]\label{eg:sample}\upshape Consider the following steps.

\begin{itemize}[itemsep=0.25em,leftmargin=*]
    \item[$\bullet$] Step \ding{182}.  Consider $s_{0}$ such that the player set $P = \{p_{1},p_{2},\ldots,p_{7}\}$ each currently holds stake $C = \{2,5,5,5,5,5,5\}$ and private valuation $\bm{\theta} = \{2,5,5,5,5,5,5\}$. However, there has been a redistribution of stake among players where $p_{1}$ has distributed a stake of $5$ to both $p_{2}$ and $p_{3}$, which means there is a directed edge from $p_{1}$ to both $p_{2}$ and $p_{3}$ with weight $5$ each. The state $s_{1} = e_{r}(s_{0})$ is the same set $P$ and $C$ without any edges between nodes.

\item[$\bullet$] Step \ding{183}. $M_{C}$ now samples randomly $(a,b) \in_{R} \{(0,1),(1,0)\}$ and sends $(s_{a},s_{b})$ to $M_{D}$. The reported valuation of each player is the same (proportional) to their allocated stake.

\item[$\bullet$] Steps \ding{192} \& \ding{193}. Based on the metric being used by $M_{D}$, it calculates $Q(s_{a})$ and $Q(s_{b})$. Now, if both values are different (w.l.o.g., $Q(s_{a}) > Q(s_{b})$) then $s_{a} \in S_{SA}$ and thus $(a,b) = (1,0)$. However, if $Q(s_{a}) = Q(s_{b})$\footnote{E.g., if the metric does not account for transaction history (represented through edges in $\mathcal{G}$) which might represent stake-redistribution among colluding parties} then $M_{D}$ cannot distinguish between $s_{a},s_{b}$ (using metric $Q(\cdot)$) better than a random guess. Without loss of generality, considering $s_{a} = s_{1}$ and $s_{b} = s_{0}$, we tabulate the values $Q(s_{a})$ and $Q(s_{b})$ for different centralization metrics -- Gini Coefficient $G_{c}$ (Eq.~\ref{eqn:gini}), Nakamoto Coefficient $S_{N}$ (Eq.~\ref{eqn:nakamoto-coefficient}), Entropy $H$ (Eq.~\ref{eqn:entropy}) and \ourmeasure\ $\Omega^{\star}$ (Eq.~\ref{eqn:cnorm}) in Table~\ref{tab:metric-comparison}.
\end{itemize}
\end{Example}

\noindent Note that with Example \ref{eg:sample}, we show scenarios where C-NORM captures the centralization while the existing metrics fail to do so. With Theorem~\ref{thm:cnorm-game}, we show that \ourmeasure\ is able to distinguish between $s_{0} \in S_{SA}$ and $s_{1} \in S_{NSA}$ for any $s_{0} \in S_{SA}$ and $s_{1} = e_{r}(s_{0})$ with probability $1 - \textsf{negl}(\kappa)$ for $\kappa$ trials. We show this by constructing a distinguisher $D_{\kappa}(\cdot)$ using centralization measure $\ourmeasure$ which distinguishes between any $s_{0} \in S_{SA}$ and corresponding $s_{1} \in S_{NSA}$. The proof is provided in Appendix~\ref{app:cnorm-game}. Example~\ref{eg:sample} and this result together show that \ourmeasure\ is an effective centralization metric. 

\begin{Theorem}\label{thm:cnorm-game}
    In $\Gamma_{cent}\langle\{M_{C},M_{D}\},S_{SA},e_{r},\Omega^{\star},\kappa\rangle$, for any $(s_{0},s_{1})$ chosen by $M_{C}$ from the set $\{(s_{0},s_{1})\;:\;s_{0} \in S_{SA} , s_{1} = e_{r}(s_{0})\}$, if $M_{D}$ uses metric \ourmeasure\ ($\Omega^{\star}$) then $Pr(D_{\kappa}(\Omega^{\star}) = 1)  > 1 - \textsf{negl}(\kappa)$ for some negligible function $\textsf{negl}(\kappa)$. Here, $S_{SA}$ is the set of all Directed Acyclic Graphs with at least one edge with non-negligible weight. %\sd{correct this statement}
\end{Theorem}

\noindent We observe as a corollary of Theorem~\ref{thm:cnorm-game} that \ourmeasure\ ($\Omega^{\star}$) also provides \emph{Sybil-proofness} because it detects any attempts of Sybil-attack and reflects it through a change in the value of $\Omega^{\star}$. 

We next show how \ourmeasure\ fits into the description of an ideal bootstrapping protocol. Towards this, in Theorem~\ref{thm:cnorm-decentralized}, we quantify decentralization (according to Definition~\ref{def:decentralization}) if we measure centralization using \ourmeasure\ $\Omega^{\star}$ of a \pos\ with an IC bootstrapping protocol $\Pi$ provided $\Omega^{\star}$ is small. The proof is in Appendix~\ref{app:cnorm-decentralized}. Note that the theorem provides sufficient (not necessary) conditions for the protocol to be $(\tau,\delta,\epsilon)-$decentralized.

\begin{Theorem}\label{thm:cnorm-decentralized}
    If a \pos\ bootstrapped using an Incentive-Compatible (IC) protocol $\Pi$ and has \ourmeasure\ value $\Omega^{\star}_{1} \leq \alpha$, then the system is $(\tau,\delta,\varepsilon)$-Decentralized for any $\delta \in [0,100]$ and $\varepsilon = \frac{2\alpha}{\bm{\beta}_{\delta}}$.
\end{Theorem}

%%%%%%%%%%%%%%%%%%%%%%%%TABLE-COMPARING-CENTRALIZATION-METRICS%%%%%%%%%%%%%%%%%%%%%%%%%%%%%%
\begin{table}[!t]
\begin{small}
\centering
\begin{tabular}{ p{15em} p{5em}  p{5em}}
\toprule
\multirow{2}{*}{\textbf{Centralization Metric}}& \multicolumn{2}{l}{\textbf{PoS Systems}}\\
 & $s_{0}$ & $s_{1}$ \\
\midrule
Nakamoto Coefficient $(N)$~\cite{centralizationMetricsGeneral2} & $3$ & $3$\\
Entropy $(H)$~\cite{entropyDecentralization} & $0.1405$ & $0.1405$ \\
Gini Coefficient $(G)$~\cite{gochhayatGini} & $0.0804$ & $0.0804$ \\
\ourmeasure\ $(\Omega^{\star})$ & $0.6$ & $0$ \\
\bottomrule
\end{tabular}
\smallskip
\caption{Scores of existing centralization metrics compared with C-NORM distinguishing between centralized $(S_{SA})$ and decentralized $(S_{NSA})$ PoS systems (refer Example~\ref{eg:sample}).}
\label{tab:metric-comparison}
\end{small}
\end{table}

Having introduced ideal bootstrapping and a centralization metric, we now analyze some of the widely used bootstrapping protocols for IC, IR, and decentralization properties. 

%%%%%%%%%%%%%%%%%%%%%%%%%%%%%%%%%%%%%%%%%%%%%%%%%%%%%%%%%%%%%
\section{Analysis of Bootstrapping Protocols}
\label{sec:analysis-existing-protocols-w2sb}
%%%%%%%%%%%%%%%%%%%%%%%%%%%%%%%%%%%%%%%%%%%%%%%%%%%%%%%%%%%%%

In this section, we (i) prove why \emph{Airdrop} and \emph{Proof-of-Burn} are not Ideal Bootstrapping protocols. (ii) formally show that a PoW-based variant (abbreviated as W2SB) is an Ideal bootstrapping protocol.

%%%%%%%%%%%%%%%%%%%%%%%%%%%%%%%%%%%%
\subsection{Airdrop \& Proof-of-Burn}
%%%%%%%%%%%%%%%%%%%%%%%%%%%%%%%%%%%%%
Airdrop and PoB are protocols widely employed in Blockchain space~\cite{PoBResource1,PoBResource2,AirdropResource1,AirdropResource2} in different capacities including bootstrapping of protocols. For Airdrop and Proof-of-Burn based bootstrapping to be ideal, these protocols must satisfy all properties outlined in Definition~\ref{def:idealbootstrappingprotocol}. Based on this, in Claim~\ref{claim:airdrop-no-ic}, we show that Airdrop is not an ideal bootstrapping protocol as it is not IC.  

\begin{Claim}\label{claim:airdrop-no-ic}
    An Airdrop-based bootstrapping protocol is not IC.
\end{Claim}

We defer the proof to Appendix~\ref{app:airdrop-no-ic} which uses the property of airdrop that 
% the payoff depends on if player $p_{i}$ is participating or not. Each player (irrespective of their reported valuation) receives the same reward (say $b_{\textsf{airdrop}}$) 
% \sd{why use $b$ as a notation}. 
% (ii) 
launching a Sybil attack does not affect utility as it does not increase $\Omega$. It is because pseudo-identities have no on-chain transactions establishing a relationship with the player launching the Sybil attack; thus, the collusion is undetected. Furthermore, to show that Proof-of-Burn (PoB) is not an ideal bootstrapping protocol, Claim~\ref{claim:pob-no-ir} proves that it does not satisfy IR. 

\begin{Claim}\label{claim:pob-no-ir}
A Proof-of-Burn-based bootstrapping protocol is not IR.
\end{Claim}

For the proof, refer to Appendix~\ref{app:pob-no-ir}, we show that a PoB bootstrapped cryptocurrency will always be upward-pegged (also referred to as one-way pegging~\cite{NarayananBitcoin}). Thus, a player's maximum utility from PoB bootstrapping in expectation is less than the utility from abstaining from the protocol.

Motivated by the switch of Ethereum from PoW to PoS~\cite{ethereumMerge}, we next discuss the potential of PoW as a bootstrapping protocol.% for PoS-based blockchains. 
% We find that PoW is a promising candidate for an ideal bootstrapping of a PoS blockchain.

\subsection{W2SB: PoW-based Bootstrapping}
Proof-of-Work (PoW) is a promising choice for a bootstrapping protocol as it offers Sybil resistance by design (Section~\ref{ssec:blockchain-prelims}). PoW involves solving a cryptographic puzzle to mine (propose) a block. Let $r_{\text{b}}$ be the reward obtained by the miner\footnote{We refer to the strategic players as miners here, as miners may be a more accessible term for PoW-based bootstrapping.} on successfully solving the puzzle. We consider that each miner $p_{i} \in P$ invests some computational resource (aka. mining power). Consider $m_{i}$ be the mining power if the invested cost is according to their true valuation (i.e. $m_{i}$ is mining power when $\hat{\theta}_{i} = \theta_{i}$). The total mining power when all miners $p_{i} \in P$ report their true valuation is $M = \sum_{i\in[n]} m_{i}$. The cost incurred per unit of mining power in a round is $\chi$ in expectation. For IC, we require $p_{i}$ gets a higher reward with mining power $m_{i}$ (corresponding to $\hat{\theta}_{i} = \theta_{i}$ ) than with $m_{i} + a$ (corresponding to $\hat{\theta}_{i} > \theta_{i}$) or $m_{i} - a$ (corresponding to $\hat{\theta}_{i} < \theta_{i}$). % \sd{<-- rewrite whatever you meant here}

\paragraph{W2SB: IC and IR.} We abstract out a Proof-of-Work based bootstrapping protocol which runs for a specified number of rounds (say $T$) and satisfies conditions stated in Lemma~\ref{lemma:pow-ic-ir} as \emph{\textbf{W}ork \textbf{to} \textbf{S}take \textbf{B}ootstrap} (\proname). We show that \proname\ is an Ideal Bootstrapping Protocol. Towards this, we first show with Lemma~\ref{lemma:pow-ic-ir} that \proname\ is both IC and IR (given some conditions on $\chi,r_{\text{b}}$ and $M$). 

\begin{lemma}\label{lemma:pow-ic-ir}
    W2SB satisfies IC and IR if (1) $\frac{\chi\cdot M}{r_{\text{b}}} \leq 1$ and (2) $\frac{\chi\cdot M}{r_{\text{b}}} \geq (1 - \frac{m_{\min}}{M})$ for $m_{\min} = \min_{i \in [n]} m_{i}$.
\end{lemma}

Appendix~\ref{app:pow-ic-ir} presents the formal proof. In the proof, for IR, we show that the expected utility of any miner, in a single round, is greater than or equal to the utility of abstaining if $\frac{\chi\cdot M}{r_{\text{b}}} \leq 1$. Next, for IC, we show that eliciting $\hat{\theta}_{i} < \theta_{i}$ and $> \theta_{i}$ both gives a lower utility when $\frac{\chi\cdot M}{r_{\text{b}}} \geq (1 - \frac{m_{\min}}{M})$.

\smallskip
\noindent\emph{Note.} Condition (1) in Lemma~\ref{lemma:pow-ic-ir} is a natural requirement for PoW and states that the mining difficulty must be such that the mining cost does not exceed the expected reward from mining. On the other hand, condition (2) requires that for any set of valuations $\bm{\theta}$, there does not exist a miner who is incentivized to increase its mining power. We remark that these are mild requirements and can be satisfied by changing $\chi$, which depends on the difficulty of the cryptographic puzzle (a tunable parameter in PoW/\proname). 

\paragraph{W2SB: Decentralization.} To show that W2SB is decentralized, we show the existence of a finite $T$ for an arbitrary $z \in \mathbb{R}_{>0}$ such that if \proname\ is run for $T$ rounds, $\Omega^{\star} < z$
% \sg{will this be $\Omega^{\star}$??}.
We assume the dynamic participation of miners is according to an arbitrary distribution $J \leftarrow \mathcal{Z}(\mu,\sigma)$. 
The CDF ($\Psi$) is defined over non-negative integers $\Psi:\mathbb{Z}_{\geq 0} \rightarrow \mathbb{R}_{\geq0}$. Likewise, the PMF($\psi$) is also defined over the set of non-negative integers $\psi:\mathbb{Z}_{\geq 0}\rightarrow \mathbb{R}_{\geq 0}$ and made discrete using the PDF $f$ as $\psi(J=q) = \int_{q-1}^{q}f(J=x)dx$. 
Under this setting, Theorem~\ref{thm:round-complexity} proves that for every $z \in (0,1]$, there exists a finite $T$ such that if we run \proname\ for $T$ rounds, we get $\Omega^{\star}\leq z$. We prove this by first upper bounding $\Omega^{\star}$ for a given round $T$. Then, we show that $\exists T$ for every $z$ such that $\Omega^{\star} \leq z$. Proof is in Appendix~\ref{app:round-complexity}.

\begin{Theorem}\label{thm:round-complexity}
    In \proname\ with miners arriving dynamically, given an arbitrary $z \in (0,1]$, there always exists a finite $T$ such that after $\geq T$ rounds, we have $\Omega^{\star} \leq z$.%\sg{instead of saying finite T, can we actually give bound?}
\end{Theorem}

 From Theorem~\ref{thm:cnorm-decentralized} and Theorem~\ref{thm:round-complexity}, we see that that \proname\ satisfies $(\tau,\delta,\epsilon)-$ Decentralization if it is run for $T$ rounds. % Corollary~\ref{cor:proto-decentralized} captures this result.
\begin{corollary}\label{cor:proto-decentralized}
    \proname\ on running for finite rounds is $(\tau,\delta,\epsilon)-$ Decentralized given miners arrive dynamically according to some distribution with CDF $\Psi(\cdot)$.
\end{corollary}

\noindent Having shown that \proname\ satisfies IR and IC (under some conditions) and is $(\tau,\delta,\epsilon)-$decentralized if run for $T$ rounds, we conclude in Theorem~\ref{thm:w2sb-ideal} that \proname\ is an ideal bootstrapping protocol.

\begin{Theorem}\label{thm:w2sb-ideal}
\proname\ is an \emph{ideal bootstrapping protocol} for $(1 - \frac{m_{\min}}{M}) \leq \frac{\chi M}{r_{\text{b}}} \leq 1$, where $m_{\min} = \min_{i \in [n]} m_{i}$.
\end{Theorem}
\begin{proof}
Lemma~\ref{lemma:pow-ic-ir} implies IC and IR, and Corollary~\ref{cor:proto-decentralized} implies $(\tau, \delta, \epsilon)-$Decentralization. Hence, \proname\ is Ideal Bootstrapping Protocol by Definition~\ref{def:idealbootstrappingprotocol} 
\end{proof}

%%%%%%%%%%%%%%%%%%%%%%%%%%%%%%%%%%%%%%%%%%%%%%%%%%%%%
\section{W2SB: Experiments \& Discussion}
\label{sec:experimental-analysis}
%%%%%%%%%%%%%%%%%%%%%%%%%%%%%%%%%%%%%%%%%%%%%%%%%%%%%%%%%%%%%%%%%%%%%%%%

We now study the empirical change in \ourmeasure\ with a change in the number of rounds, $T$. From Theorem~\ref{thm:round-complexity}, we know that \proname\ reaches an arbitrary level of decentralization (reflected by a decrease in \ourmeasure) for a sufficient $T$. Our experiments aim to quantify \ourmeasure\ in a PoS system when W2SB is run for different values of $T$. We begin by explaining our experimental setup, followed by the results and their discussion. 
%%%%%%%%%%%%%%%%%%%
%
%%%%%%%%%%%%%%%%%%%%%%%%%%%%%%%%%%%%
%%%%%%%%%%%%%%%%%%%%%%%%%%%%%%%%%%%%%%%%%%%%%%%%%%%%%
\begin{figure}[!t]
    \centering
  \begin{subfigure}{0.7\columnwidth}
    \includegraphics[width=0.7\columnwidth]{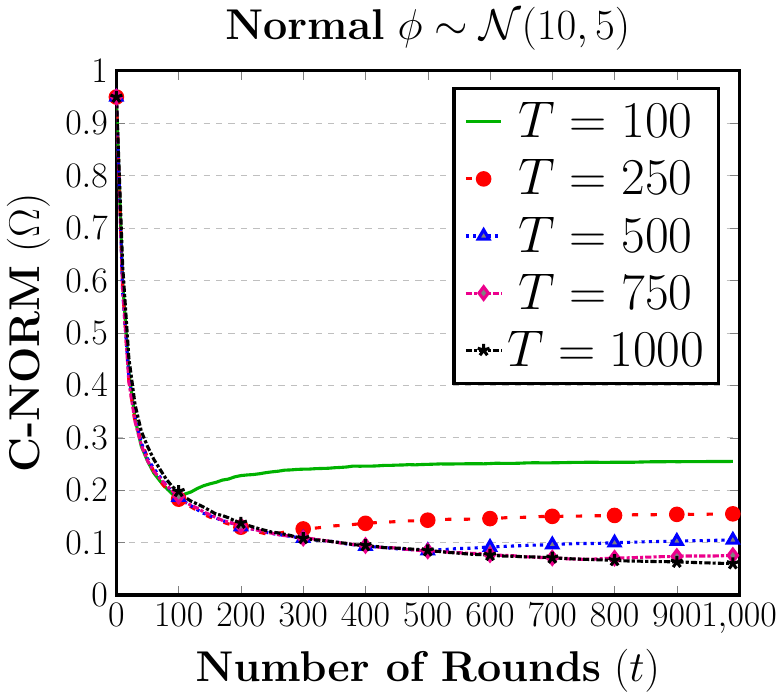}
  \end{subfigure}%
  %\hspace*{\fill}   % maximize separation between the subfigures
  \begin{subfigure}{0.2\columnwidth}
        % \begin{small}
        \centering
        \begin{tabular}{ p{2em} | p{2em}}
        \toprule
        $z$ & $T$\\
        \midrule
        $0.4$ & $21$ \\
        $0.3$ & $35$ \\
        $0.2$ & $82$ \\ 
        $0.1$ & $350$\\
        \bottomrule
        \end{tabular}
        % \smallskip
        % \end{small}
        % \caption{Number of rounds required, $T$, for varying $z$.}
      \end{subfigure}%

    \centering
    \caption{\ourmeasure\ against rounds for different stopping time $T$ for distribution $\mathcal{N}(7,3)$}
    \label{fig:results}     
\end{figure}

%%%%%%%%%%%%%%%%%%%%%%%%%%%%%%%%%%%%%%%%%%%%%%%%%%%%
\paragraph{Experimental Setup.}
We simulate a PoS-based blockchain bootstrapped using \proname. To empirically measure \ourmeasure, we run \proname\ for different stopping times $T \in \{100,200,500,750,1000\}$, after which the PoS protocol is run for remaining $(1000-T)$ rounds.

To capture each participant $i$'s dynamic arrival, we define the random variable $\mathcal{J}_{i,t}$, which takes the value $1$ if $i$ joins the system at round $t$ and 0 otherwise. We model $\mathcal{J}_{i,t} \sim Chi(k=3)$ from a chi-squared distribution with degree of freedom $k = 3$. We believe this distribution aptly exhibits participant arrival: sudden increase at first, which peaks closer to the starting, followed by a gradual decrease and long tail. We simulate three different instances where each player $p_{i}$ samples its mining power $\phi_{i}$ (proportional to $\theta_{i}$) from (1) Gaussian $\phi\sim\mathbf{\mathcal{N}}(7,3)$, (2) Uniform $\phi\sim\mathbf{\mathcal{U}}(0,50)$, and (3) Exponential $\phi\sim\mathbf{Exp}(\lambda=1)$ distributions.\footnote{We remark that the general trend presented remains similar for different distribution parameters (refer Appendix~\ref{app:plots}).} %\sd{remember to add these additional plots} 

\paragraph{Results \& Inference.}
Figure~\ref{fig:results} depicts our results. We make the following key observations.

\begin{enumerate}[itemsep=0pt]
    \item \textbf{In W2SB $\Omega^{\star} \to 0$ as $t\to \infty$. } From Figure~\ref{fig:results}, we observe that if the W2SB runs for a sufficiently large $T$, then \ourmeasure\ tends to zero, i.e., $\Omega^{\star} \to 0$. 

    \item \textbf{\ourmeasure\ does not decrease after Bootstrapping. } We observe that for rounds $> T$, \ourmeasure\ saturates to a fixed value. This is because the PoS protocol does not cause any change in \ourmeasure\ of the protocol, as desired. 
    % As evident in Figure~\ref{fig:results}, \ourmeasure\ decreases till W2SB is run, after which it saturates depending on the stopping time $T$. 

    \item \textbf{\proname\ achieves $\Omega^{\star} \leq z$ for any $z > 0$. } Any $z$ can be achieved in a finite number of rounds in W2SB, as $\Omega^{\star}$ tends to $0$ asymptotically. The smaller value of $z$, the larger number of rounds required. E.g., the green curve in Figure~\ref{fig:results} with small $T$ is relatively centralized. %\sd{should it be $\delta$ or $z$ from Theorem~\ref{thm:round-complexity}?}
\end{enumerate}

\if 0
\begin{table}[!t]
\begin{small}
\centering
\begin{tabular}{ p{3em} | p{7em} p{7em} p{7em}  }
\toprule
\multirow{2}{*}{$z$}& \multicolumn{3}{c}{\textbf{Stopping Round}$(T)$}\\
 & $\phi\sim \mathcal{N}(7,3)$ & $\phi\sim \mathcal{U}[0,50]$ & $\phi\sim \textbf{Exp}(\lambda=1)$ \\
\midrule
$0.4$ & $21$ & $25$ & $44$ \\
$0.3$ & $35$ & $44$ & $88$ \\
$0.2$ & $82$ & $95$ & $223$ \\ 
$0.1$ & $350$ & $393$ & $\approx 1000$ \\
\bottomrule
\end{tabular}
\smallskip
\caption{Number of rounds required, $T$, for varying $z$.}
\label{tab:trends}
\end{small}
\end{table}
\fi
%
%%%%%%%%%%%%%%%%%%%%%%%%%%%%%%%%%
\paragraph{Discussion.}\label{ssec:discussion-future}
%%%%%%%%%%%%%%%%%%%%%%%%%%%%%%%%%
Our results (both theoretical and empirical) show that \proname\ is an Ideal bootstrapping protocol that achieves sufficiently small \ourmeasure\ values. The number of rounds for which \proname\ is run ($T$) is a hyper-parameter that is decided by protocol descriptor, allowing a tradeoff between higher levels of decentralization (by increasing $T$) and lower energy consumption (by decreasing $T$). 

\noindent \textbf{Ethereum Merge -- an instance of \proname.} Ethereum also used PoW initially before switching to PoS (aka the Merge~\cite{ethereumMerge}). Although the intention of the Merge might not have been decentralized stake distribution, the consequence was that Ethereum became a decentralized PoS-based blockchain. 
% However, the number of rounds the PoW ran for was significant, which from Lemma~\ref{lemma:round-complexity} could be bounded as $\bm{\theta}(\ln({1}/{\delta}))$ for any desirable $\delta$ centralization. 

%%%%%%%%%%%%%%%%%%%%%%%%%%%%%%%%%%%%%%%%%%%%%%%%%%%%%%%%%%%%%%%%%%%%%%%%
%%%%%%%%%%%%%%%%%%%%%%%%%%%%%%%%%%%%%%%%%%%%%%%%%%%%%
\section{Conclusion \& Future Work}

\paragraph{Conclusion.} This work attempted to resolve the problem of centralization during bootstrapping of PoS-based blockchains. Towards this, we first presented a game-theoretic model of bootstrapping, $\Gamma_\textsf{bootstrap}$. With this, we defined an ideal bootstrapping protocol, i.e., a protocol that simultaneously satisfies IR, IC, and decentralization. To quantify centralization, we introduced \ourmeasure\, which measures centralization in PoS-based blockchains. We show the effectiveness of \ourmeasure\ against existing metrics, using the centralization game, $\Gamma_{cent}$. We then analyzed existing protocols under our model to show that while Airdrop and Proof-of-Burn are not IC and IR, respectively, PoW-based bootstrapping (\proname) is an  Ideal bootstrapping protocol. Our work lays the theoretical foundations for further analysis of PoS-based blockchains.

\paragraph{Future Work.} Decentralized bootstrapping in PoS-based blockchains is relatively new. One potential future direction is to explore verifiable random functions \cite{vrf} for greater energy efficiency than \proname. Another challenge with PoS-based blockchains is that centralization may occur due to stake pools post the bootstrapping phase. Tackling post-bootstrapping threats to decentralization may be an interesting future direction to explore.

%%%%%%%%%%%%%%%%%%%%%%%%%%%%%%%%%%%%%%%%%%%%%%%%%%%%%%%%%%%%%%%%%%%%%%%%

%%% The next two lines define, first, the bibliography style to be 
%%% applied, and, second, the bibliography file to be used.

% \newpage
\bibliographystyle{ACM-Reference-Format} \bibliography{references}

% \newpage
\appendix 

% \section{Centralization Game}

% \begin{itemize}
%     \item Challenger:
%     \begin{itemize}
%         \item Sample a random bit $b \leftarrow \{0,1\}$  
%         \item If $b=0$:
%             \begin{itemize}
%                 \item Corresponds to a centralized PoS System
%             \end{itemize}
%         \item If $b=1$:
%             \begin{itemize}
%                 \item Corresponds to a decentralized PoS System
%             \end{itemize}
%         \item Provide the attacker $\mathcal{G}$ \sd{other stuff that define the problem}
%         \item Provide the attacker with a threshold value for centralization, i.e., $\Omega^{\star}_\textsf{threshold}$
%     \end{itemize}
%     \item Attacker:
%         \begin{itemize}
%             \item Has a centralization metric \textsf{CM} 
%             \item Runs \textsf{CM} on $\mathcal{G}$
%             \item If $\textsf{CM}(\mathcal{G})>\Omega^{\star}_\textsf{threshold}$
%                 \begin{itemize}
%                     \item Attacker selects $b=0$
%                 \end{itemize}
%             \item Else
%                 \begin{itemize}
%                     \item Attacker selects $b=1$
%                 \end{itemize}
%         \end{itemize}
%     \item The attacker's success is the number of times it correctly predicts the bit
% \end{itemize}

\section{Proof for Theorem~\ref{thm:cnorm-game}}
\label{app:cnorm-game}
\begin{proof}
    To show that \ourmeasure\ can distinguish between $s_{0}$ and $s_{1}$ we consider any (arbitrary) $s_{0} \in S_{SA}$ and corresponding $s_{1} \in S_{NSA}$. We show for each such sampling for $s_{0},s_{1}$ we have $\Omega^{\star}(s_{0}) \neq \Omega^{\star}(s_{1})$ which concludes our proof. 

    We are considering the experiment being conducted when the stake distribution is happening through an Incentive-Compatible bootstrapping protocol. Therefore, $\hat{\theta}_{i} = \theta_{i}$ for any player $p_{i}$.
    Consider two samples $s_{0} \in S_{SA}$ and $s_{1} = e_{r}(s_{0})$. The goal is for $\Omega^{\star}$ to distinguish between $(s_{0},s_{1})$. 
    
    First we calculate $\Omega^{\star}$ for $s_{1}$. Since there are no edges, $\omega_{i} = c_{i}$. According to the bootstrapping protocol, the allocated stake $c_{i}$ is proportional to the valuation $\theta$ for a player $p_{i}$. Therefore, we get $\omega_{i}/\hat{\theta}_{i} = \theta_{i}\cdot z/\hat{\theta}_{i}$ for some constant $z > 0$. However, due to the Incentive Compatible property of the bootstrapping protocol, we have $\omega_{i}/\hat{\theta}_{i} = z$. Thus, scaled stake $\beta_{i} = \frac{z}{\sum_{i=1}^{n}z} = \frac{1}{n}$ which gives $\Omega^{\star} = 0$.

    Now, to distinguish from the other case, we want $\Omega^{\star} > 0$ (by more than negligible value) for $s_{0}$. Consider $s_{0} \in S_{SA}$. There is atleast one (non-negligible weighted) edge and the graph is a Directed Acyclic Graph, which means there is at least one sink node (node with only incoming edges). Wlog. let this node correspond to player $p_{i}$. The sum of weights of incoming edges be $q$. Therefore, $\omega_{i}/\hat{\theta}_{i} = (c_{i} + q)/\hat{\theta}_{i} = (z\cdot\theta_{i})/\hat{\theta_{i}} + q/\hat{\theta}_{i} = z + q/\theta_{i}$. The scaled stake $\beta_{i} = \frac{z + q/\theta_{i}}{\sum_{j=1}^{n} \omega_{j}/\hat{\theta}_{j}} \geq \frac{z + q/\theta_{i}}{nz} = \frac{1}{n} + \frac{q}{n\cdot\theta_{i}}$. Now we get $\Omega^{\star}$ as 
    \[
    \begin{aligned}
        \Omega^{\star} & = & \frac{1}{2}\sum_{j=1}^{n} \Big|\beta_{j} - \frac{1}{n} \Big| & \geq & \beta_{i} - \frac{1}{n} & = & \frac{q}{n\cdot\theta_{i}}
    \end{aligned}
    \]
    Therefore, $\Omega^{\star} \geq \frac{q}{n\cdot\theta_{i}}$ which is a non-negligible value. Therefore, for any pair $(s_{0},s_{1})$ we can distinguish between the two states. Let the player $M_{D}$ be given $(s_{a},s_{b})$ such that (wlog.) $\Omega^{\star} = 0$ for $s_{a}$ and $\Omega^{\star} > 0$ for $s_{b}$. Then $a = 1, b=0$. Similarly, $M_{D}$ can also predict correctly in case of $a = 0,b = 1$. If we are repeating this $\kappa$ times, then $M_{D}$ correctly predicts the ordering of $(a,b)$ with probability $1 - \textsf{negl}(\kappa)$.
\end{proof}

\section{Proof for Theorem~\ref{thm:cnorm-decentralized}}
\label{app:cnorm-decentralized}

\begin{proof} Consider $\Omega^{\star}_{1} = \alpha$. The set of players $P = \{p_{1},p_{2}\ldots,p_{n}\}$ and \emph{scaled stake} is $\bm{\beta} = \{\beta_{1},\beta_{2},\ldots,\beta_{n}\}$. Wlog. we consider $\beta_{1} \geq \beta_{2} \geq \ldots \geq \beta_{n}$. Therefore, $\beta_{max} = \beta_{1}$. Now, we can write $\Omega^{\star}_{1}$ as
\[
\begin{aligned}
\frac{1}{2}\sum_{i=1}^{n} \Big|\beta_{i} - \frac{1}{n}\Big| = & \;\alpha \\
\end{aligned}
\]
We consider $\delta^{th}$ percentile of $\bm{\beta}$ as the player $p_{i}$ which has the next lowest stake to $\delta$ percentile of players. Let $\bm{\beta}_{\delta}$ be used to represent this $\delta^{th}$ percentile. Therefore, $\bm{\beta}_{\delta} := \beta_{\lceil(1- \delta)\cdot n\rceil}$. 

\smallskip \noindent \underline{Case 1: $\bm{\beta}_{\delta} \leq \frac{1}{n}$:} In this case, we can write $\Omega^{\star}_{1} = \alpha$ as,
\[
\begin{aligned}
    \frac{1}{2}\left(\Big|\beta_{max} - \frac{1}{n}\Big| + \Big|\bm{\beta}_{\delta} - \frac{1}{n}\Big|\right) & \leq & \alpha \\
    \beta_{max} - \frac{1}{n} + \frac{1}{n} - \bm{\beta}_{\delta} & \leq & 2\alpha \\ 
    \frac{\beta_{max}}{\bm{\beta}_{\delta}} & \leq & 1 + \frac{2\alpha}{\bm{\beta}_{\delta}}
\end{aligned}
\]
\smallskip\noindent\underline{Case 2: $\bm{\beta}_{\delta} \geq \frac{1}{n}$:}
In that case, we can write $\Omega^{\star}_{1} = \alpha$ as 
\[
\begin{aligned}
    \beta_{max} - \frac{1}{n} & \leq & 2\alpha \\ 
    \frac{\beta_{max}}{\bm{\beta}_{\delta}} & \leq & \frac{1/n}{\bm{\beta}_{\delta}} + \frac{2\alpha}{\bm{\beta}_{\delta}} & \leq & 1 + \frac{2\alpha}{\bm{\beta}_{\delta}}
\end{aligned}
\]
The last inequality comes since $\bm{\beta}_{\delta} \geq \frac{1}{n}$. We have therefore shown that for $\Omega^{\star}_{1} = \alpha$, the system satisfies the Proportionality condition. Additionally, $\Omega^{\star}$ (\ourmeasure) captures attempts of sybil attacks successfully (as demonstrated from Theorem~\ref{thm:cnorm-game}), and we can enforce on the protocol the condition for \emph{Minimum Participation}. Therefore, if $\Omega^{\star}_{1} = \alpha$ then protocol (which ensures \emph{Minimum Participation}) is $(\tau,\delta,\frac{2\alpha}{\bm{\beta}_{\delta}})-$Decentralized for any $\delta \in [0,1]$. 
\end{proof}

\section{Proof for Claim~\ref{claim:airdrop-no-ic}}
\label{app:airdrop-no-ic}
\begin{proof} In proving non-IC property for Airdrop, we show that if a party forms pseudo-identities, they can always obtain a higher utility than honestly disclosing their valuations, even when other players are reporting their true valuations and are not a part of any coalition. This is because, in Airdrop, all players get the same reward irrespective of their valuation, which incentivizes them to split their valuation among pseudo-identities to obtain a higher utility.
Consider an Indicator Function $\bm{1}_{\hat{\theta}_{i}>0}$ which is $1$ if $\hat{\theta}_{i} > 0$ else $0$. The utility function for \emph{airdrop} bootstrapping protocol for player $p_{i}$ is written as
\[
\begin{aligned}
    U_{i}(\hat{\theta}_{i},\bm{\hat{\theta}}_{-i},\emptyset,\emptyset;\theta_{i}) & = & b_{airdrop}\cdot\bm{1}_{\hat{\theta}_{i}>0} - \Omega^{\star}\cdot g(\theta_{i})
\end{aligned}
\]
For Nash Incentive Compatibility (IC) from Equation~\ref{eqn:ic} of Definition~\ref{def:ic} we require $\forall\;\hat{\theta}_{i} \in \mathbb{R}_{\geq 0},\forall\;A_{i} \in \mathcal{A}_{i}$
\[
\begin{aligned}
    U_{i}(\theta_{i},\bm{\theta}_{-i},\emptyset,\emptyset;\theta_{i}) & \geq & U_{i}(\hat{\theta}_{i},\bm{\theta}_{-i},A_{i},\emptyset;\bm{\theta}) \\ 
    b_{airdrop}\cdot\bm{1}_{\hat{\theta}_{i}>0} - \Omega^{\star}\cdot g(\theta_{i}) & \geq & b_{airdrop}\cdot\sum_{j \in A_{i}}\bm{1}_{\hat{\theta}_{j}>0} - \Omega^{\star}\cdot g(\theta_{i})
\end{aligned}
\]
We therefore observe that by forging identities (Sybil-attack) which is equivalent to forming a coalition $A_{i}$ we observe that a player can forge arbitrary number of identities and gain more reward than following the protocol honestly. This means for any $A_{i}$ such that $|A_{i}| > 1$ the above inequality does not hold true. Thus, Airdrop does not satisfy Nash Incentive Compatibility (IC). 
\end{proof}

\section{Proof for Claim~\ref{claim:pob-no-ir}}
\label{app:pob-no-ir}
\begin{proof}
     Consider a Proof-of-Burn-based bootstrapping where cryptocurrency from an \emph{Old Crypto Token} $\$OCT$ is burnt to obtain \emph{New Crypto Token} $\$NCT$. The conversion rate from $OCT$ to $USD$ is $1\;OCT = d\;USD$ and for $NCT$ is $1\;NCT = e\;USD$. The rule is set such that if a player burns $a\;OCT$ they obtain $b\;NCT$. Payoff on not participating in the protocol for a player $p_{i}$ is $U_{i}(0,\bm{\theta}_{-i},\emptyset,\emptyset;\bm{\theta}) = -\Omega^{\star}\cdot g(\theta_{i})$. The payoff obtained from participating in the protocol honestly, for some constant $\gamma \in \mathbb{R}_{> 0}$ is 
     \[
    \begin{aligned}
        U_{i}(\theta_{i},\bm{\theta}_{-i},\emptyset,\emptyset;\theta_{i}) & = & (b\cdot e - a\cdot d)\gamma\theta_{i} - \Omega^{\star}\cdot g(\theta_{i})
    \end{aligned}
     \]
     However, setting an exchange rate for the New Crypto Token (which is set by the protocol designers) does not establish the exchange rate, but sets a ``one-way peg'' or price-ceiling for $NCT$ such that $b\cdot e \leq a\cdot d$ (See~\cite{NarayananBitcoin} Section~10.1 for more details). If we account for transaction-fees and expected utility, then we get 
     \[
     \begin{aligned}
         \mathbb{E}[U_{i}(\theta_{i},\bm{\theta}_{-i},\emptyset,\emptyset;\theta_{i})] & = \mathbb{E}[(b\cdot e - a\cdot d)\gamma\theta_{i} - \Omega^{\star}\cdot g(\theta_{i})]\\
         & < -\mathbb{E}[\Omega^{\star} g(\theta_{i})] = \mathbb{E}[U_{i}(0,\bm{\theta}_{-i},\emptyset,\emptyset;\bm{\theta})]
     \end{aligned}
     \]
     Therefore, PoB-based bootstrapping protocol is not Individually Rational (IR).
\end{proof} 

\section{Proof for Lemma~\ref{lemma:pow-ic-ir}}
\label{app:pow-ic-ir}
\begin{proof}
    
Consider a PoW blockchian protocol such that if all players invest in mining proportionally to their true valuation then player $p_{i}$ with valuation $\theta_{i}$ has mining power $m_{i} = \theta_{i}\cdot l$ (for some $l \in \mathbb{R}_{>0}$). The cost incurred per unit mining power for any player is $\chi$ and the reward obtained on mining a block be $r_{\text{b}}$.The minimum mining power which a player can have be $m_{min}$. This means $\forall p_{i} \in P, m_{i} \geq m_{min}$. Additionally, let $M = \sum_{i \in [n]} m_{i}$.

Consider for player $p_{i}$ a random variable $R_{i}$ which denotes the payoff of the player. 
\[
    R_{i}= 
\begin{cases}
    r_{\text{b}} - \chi m_{i},& \text{with probability } \frac{m_{i}}{M}\\
    -\chi m_{i},              & \text{with probability} \frac{M - m_{i}}{M}
\end{cases}
\]

\noindent\smallskip \underline{Individual Rationality:} For IR, we require that expected block reward should exceed cost of mining for all players $p_{i}$. Therefore $\forall p_{i} \in P$ the expected utility on abstaining from the protocol is 
\[
    \begin{aligned}
        \mathbb{E}[U_{i}(0,\bm{\hat{\theta}}_{-i},\emptyset,A_{-i};\theta_{i})] & = - \mathbb{E}[\Omega^{\star}\cdot g(\theta_{i})]
    \end{aligned}
\]
The expected utility on participating in the protocol for each round is
\[
    \begin{aligned}
        \mathbb{E}[U_{i}(\theta_{i},\bm{\theta}_{-i},\emptyset,\emptyset;\theta_{i})] & = \mathbb{E}[R_{i}] - \mathbb{E}[\Omega^{\star}\cdot g(\theta_{i})]\\
        & = (r_{\text{b}}\frac{m_{i}}{M} - \chi m_{i}) - \mathbb{E}[\Omega^{\star}\cdot g(\theta_{i})]
    \end{aligned}
\]
For IR we require $\mathbb{E}[U_{i}(0,\bm{\hat{\theta}}_{-i},\emptyset,A_{-i};\theta_{i})] \leq \mathbb{E}[U_{i}(\theta_{i},\bm{\theta}_{-i},\emptyset,\emptyset;\theta_{i})]$ which gives us $\forall p_{i} \in P$
\[
    \begin{aligned}
        (r_{\text{b}}\frac{m_{i}}{M} - \chi m_{i}) & \geq 0 \\
        r_{\text{b}}\frac{m_{i}}{M} & \geq \chi m_{i} & \Rightarrow \frac{\chi M}{r_{\text{b}}} \leq 1\\
    \end{aligned}
\]

\noindent\smallskip \underline{Incentive Compatibility:} For Incentive Compatibility, we will show for each player $p_{i} \in P$, reporting a valuation $\hat{\theta}_{i} \neq \theta_{i}$ will lead to a lower payoff. We proceed in two cases:

\noindent \underline{Case 1 $\hat{\theta}_{i} < \theta_{i}$:} In this case, let $m_{i}^{'}$ be the mining power corresponding to the disclosed valuation $\hat{\theta}_{i}$. Clearly, $m_{i}^{'} < m_{i}$. Consider for some $a > 0, m_{i}^{'} = m_{i} - a$. The expected difference in utility, that is $\mathbb{E}[U_{i}(\hat{\theta}_{i},\bm{\theta}_{-i},A_{i},\emptyset;\theta_{i}) - U_{i}(\theta_{i},\bm{\theta}_{-i},\emptyset,\emptyset;\theta_{i})]$ is therefore given by 
\[
    \begin{aligned}
        \mathbb{E}[U_{i}(\hat{\theta}_{i},\bm{\theta}_{-i},A_{i},\emptyset;\theta_{i}) - U_{i}(\theta_{i},\bm{\theta}_{-i},\emptyset,\emptyset;\theta_{i})] \\ = r_{\text{b}}\left(\frac{m_{i}^{'}}{M - m_{i} + m_{i}^{'}} - \frac{m_{i}}{M}\right) - \chi(m_{i}^{'} - m_{i}) \\
        = \frac{a}{M-a}\left((M - a)\chi - r_{\text{b}}\right) \\
        < \frac{a}{M - a}\left(M\chi - r_{\text{b}}\right) \leq 0
    \end{aligned}
\]
The last inequality comes from the condition obtained from Individual Rationality $\chi M \leq r_{\text{b}}$.

\noindent\underline{Case 2 $\hat{\theta}_{i} > \theta$:} In this case, let $m^{'}_{i} = m_{i} + a$ (for some $a > 0$) be the mining power corresponding to reported valuation $\hat{\theta}_{i}$. Consider $R^{'}_{i}$ be the random variable denoting reward when reported valuation is $\hat{\theta}_{i}$. Therefore, 
\[
    R_{i}^{'}= 
\begin{cases}
    r_{\text{b}} - \chi(a + m_{i}),& \text{with probability } \frac{m_{i}+a}{M+a}\\
    -\chi(a + m_{i}),              & \text{with probability} \frac{M - m_{i}}{M + a}
\end{cases}
\]
Now, we can write the expected difference in utility as 
\[
    \begin{aligned}
        \mathbb{E}[U_{i}(\hat{\theta}_{i},\bm{\theta}_{-i},A_{i},\emptyset;\theta_{i}) - U_{i}(\theta_{i},\bm{\theta}_{-i},\emptyset,\emptyset;\theta_{i})] & = \mathbb{E}[R^{'}_{i} - R_{i}]\\
         = r_{\text{b}}\left(\frac{m_{i} + a}{M + a} - \frac{m_{i}}{M}\right) - \chi (m_{i} + a - m_{i})\\
         = r_{\text{b}}\frac{a(M - m_{i})}{M(M + a)} - a\chi \\ 
         = \frac{a r_{\text{b}}}{M+a}\left(1 - \frac{m_{i}}{M} - \frac{\chi(M + a)}{r_{\text{b}}}\right) \\
         \leq \frac{a r_{\text{b}}}{M+a}\left(1 - \frac{m_{min}}{M} - \frac{M\chi}{r_{\text{b}}}\right) < 0
    \end{aligned}
\]
The last inequality uses the fact $\left(1 - \frac{m_{min}}{M}\right) < \frac{M\chi}{r_{\text{b}}}$. 

In conclusion, we have shown that if $1 - \frac{m_{min}}{M} < \frac{M\chi}{r_{\text{b}}} \leq 1$ is satisfied, then the PoW-based bootstrapping protocol (W2SB) is both Individually Rational (IR) and Incentive Compatible (IC). 
\end{proof}

\begin{figure}[!th]
    \centering
    \includegraphics[width=\linewidth]{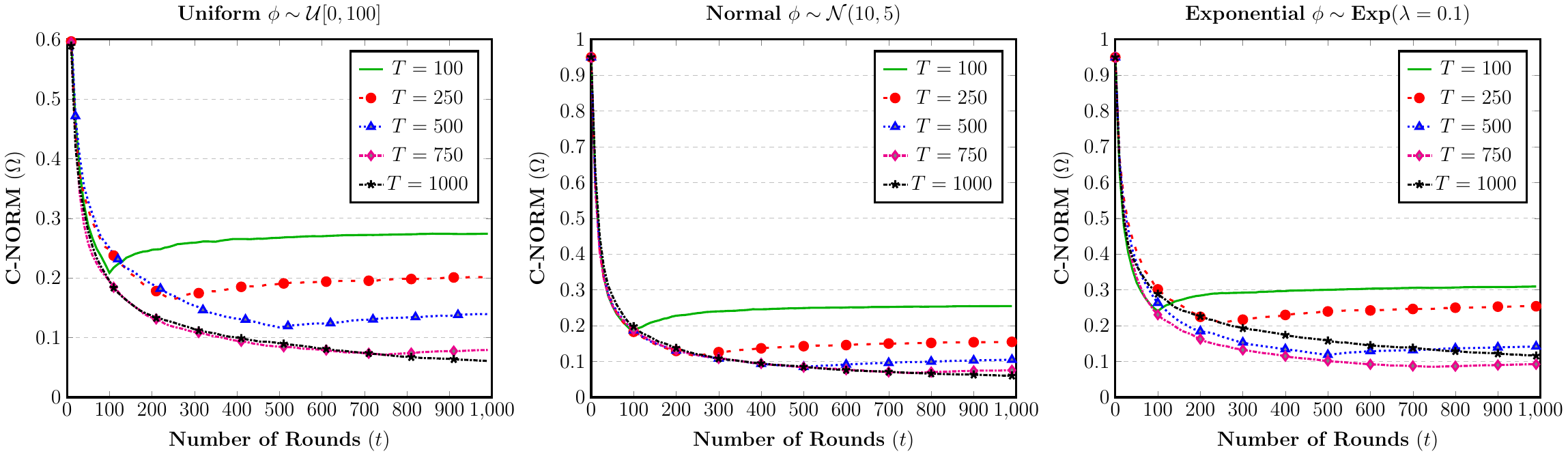}
    \caption{\ourmeasure\ against rounds for different $T$ and different distributions of miner stake}
    \label{fig:appendix-plots}     
\end{figure}

\section{Proof for Theorem~\ref{thm:round-complexity}}
\label{app:round-complexity}
\begin{proof}
To prove the theorem, we run the \proname\ for $T$ rounds. We show that for any arbitrary desired value $z \in (0,1)$ there always exists $T$ such that after $T$ rounds, $\Omega^{\star} \leq z$. Out of these $T$ rounds, the first $T_{0} < T$ rounds are such that the majority of players have joined the system before this round. This means $T_{0}$ is such that $\Psi(T_{0}) = 1 - \frac{x}{n}$ for $x < z$. Wlog. we assume players join the system in the order $p_{1},p_{2},\ldots,p_{n}$. The \ourmeasure\ value is therefore given by
\[
    \begin{aligned}
        \Omega^{\star} & = \frac{1}{2}\sum_{i=1}^{n} \Big|\beta_{i} - \frac{1}{n}\Big|\\
        & = \frac{1}{2}\sum_{i=1}^{\Psi(T_{0})n} \Big|\beta_{i} - \frac{1}{n}\Big| + \sum_{i=\Psi(T_{0})n}^{n} \Big|\beta_{i} - \frac{1}{n}\Big| \\
        & \leq \sum_{i=1}^{\Psi(T_{0})n} \Big|\beta_{i} - \frac{1}{n}\Big| + \sum_{i=\Psi(T_{0})n}^{n} 1 \\ & = \sum_{i=1}^{\Psi(T_{0})n} \Big|\beta_{i} - \frac{1}{n}\Big| + n\cdot(1 - \Psi(T_{0}))
    \end{aligned}
\]
We get the last inequality because $|a - b| \leq \max(|a|,|b|)$ and $|\beta_{i} - \frac{1}{n}| \leq \max(|\beta_{i}|,\frac{1}{n}) \leq 1$. Now we consider $\Big|\beta_{i} - \frac{1}{n}\Big|$, which gives us 
\[
    \begin{aligned}
        \Big|\beta_{i} - \frac{1}{n}\Big| & = \Big|\frac{\omega_{i}/\theta_{i}}{\sum_{j=1}^{n}\omega_{j}/\theta_{j}} - \frac{1}{n}\Big|\\
        & \leq \Big|\frac{T_{0}\cdot r_{\text{b}} + (\omega^{'}_{i}/\theta_{i})}{\sum_{j=0}^{n} \chi(T - T_{0})} - \frac{1}{n}\Big|\\
    \end{aligned}
\]
For the last inequality, (1) $\omega_{i}^{'}$ is the payoff from round $T_{0}$ to $T$, and before that any player can get at most $T_{0}r_{\text{b}}$. Therefore, we use the inequality $\omega_{i}/\theta_{i} \leq T_{0} + \omega_{i}^{'}/\theta_{i}$ to upper bound the numerator. (2) $\omega_{j}/\theta_{j} \geq \omega_{j}^{'}/\theta_{j} \geq (T - T_{0})\chi$ since \proname\ is IR means $\frac{\theta_{i}r_{\text{b}}}{\sum_{j=1}^{n}\theta_{j}} \geq \chi$ and therefore $\omega_{j}/\theta_{j} \geq \chi(T - T_{0})$ as the stake is allocated for $T - T_{0}$ rounds. 

We can also upper bound $\omega_{i}^{'}/\theta_{i}$ by considering the stake is distributed among only the players who have joined in round $\leq T_{0}$. In this case, $\omega_{i}^{'}/\theta_{i} = c\;\forall\;p_{i}$ joining before round $T_{0}$ for some constant $c$. We therefore get 
\[
    \begin{aligned}
        \Omega^{\star} & \leq n\Psi(T_{0})\Big|\frac{T_{0}r_{\text{b}} + (T - T_{0})c}{n\chi(T - T_{0})} - \frac{1}{n}\Big| + x\\
        & = n\Psi(T_{0})\Big|\frac{c}{n\chi} + \frac{T_{0}r_{\text{b}}}{n\chi(T - T_{0})} - \frac{1}{n}\Big| + x
    \end{aligned}
\]
When we increase $T$, we can reduce the mod term to a small enough value such that $\Omega^{\star} \leq z$. By IR we require $c = r_{\text{b}}/M \approx \chi/n$ (for very small $m_{min}$, see Lemma~\ref{lemma:pow-ic-ir}). Therefore, 
\[
    \begin{aligned}
        \Omega^{\star} &\lesssim \Psi(T_{0})\Big|\frac{T_{0}r_{\text{b}}}{\chi(T - T_{0})}\Big| + x &\leq z
        & \Rightarrow T \geq \frac{\Psi(T_{0})T_{0}r_{\text{b}}}{(z - x)\chi} + T_{0}
    \end{aligned}
\]

Thus, for any $z \in (0,1]$ we can always obtain a finite $T$ such that $\Omega^{\star} \leq z$ if \proname\ is run for $\geq T$ rounds, assuming dynamic participation according to some  distribution with CDF $\Psi$.
\end{proof}
\section{Experiments}
\label{app:plots}
The plots are shown in figure~\ref{fig:appendix-plots} for distributions with different parameters to show that the trend remains the same.

\section{Cycle Elimination Algorithm}
\label{app:cycle-elimination}
%%% MOVE THIS TO APPENDIX %%%%%%%%%%
\paragraph{Cycle Elimination Procedure.} Without loss of generality, consider there exists a cycle with edges from $p_{1}$ to $p_{2}$, $p_{2}$ to $p_{3}$ and so on till $p_{k}$ to $p_{1}$. These cycles can be found using any cycle detection algorithm such as BFS, Floyd's algorithm. Let weight $w_{1,2}$ be the smallest of the weights. We eliminate the edge from $p_{1}$ to $p_{2}$ by subtracting the weight $w_{1,2}$ from each edge of the cycle. If there exist multiple cycles, the order of elimination will result in different DAGs. However,  the resultant value of the centralization metric (proposed in Section~\ref{ssec:ourmeasure}) does not change.  

\begin{figure}[t]
\begin{small}
  \begin{subfigure}{0.25\linewidth}
    \includegraphics[width=\linewidth, clip]{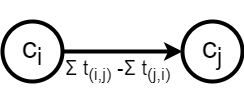}
    \caption{Edges} \label{fig:graph1}
  \end{subfigure}%
  \hspace*{2em}   % maximize separation between the subfigures
  \begin{subfigure}{0.25\linewidth}
    \includegraphics[width=\linewidth,clip]{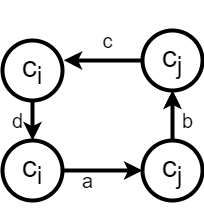}
    \caption{With cycles} \label{fig:graph2}
  \end{subfigure}%
  \hspace*{2em}   % maximize separation between the subfigures
  \begin{subfigure}{0.25\linewidth}
    \includegraphics[width=\linewidth, clip]{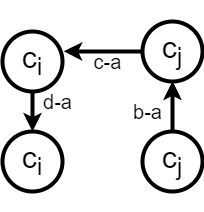}
    \caption{Eliminated Cycle} \label{fig:graph3}
  \end{subfigure}%
\caption{PoS System graph and cycle elimination.} \label{fig:graph-figure}
\end{small}
\end{figure}

\begin{algorithm}[th]
\small
\caption{\texttt{CycleElimination} procedure}\label{alg:cycle-elimination}
\begin{algorithmic}[1]
\Statex\textbf{Input:} $\mathcal{G} = (P,\mathcal{C},\mathcal{W})$ %%% T, \theta 
\color{blue}\Comment{PoS System Graph (might contain cycles)}\color{black}
\Statex\textbf{Output:} $\mathcal{G}^{'}$
\color{blue}\Comment{Directed Acyclic PoS System Graph}\color{black}
\For{each cycle $(p_{1},p_{2},\ldots,p_{k})$}
\State $w_{min} := \infty$
\For{each $p_{i},p_{j}$ in the cycle}
\State $w_{min} = min(w_{min}, w_{i,j})$
\EndFor
\For{each $p_{i},p_{j}$ in the cycle}
\State $w_{i,j} = w_{i,j} - w_{min}$
\EndFor
\EndFor
\State Updated weight set is $\mathcal{W}^{'}$
\State $\mathcal{G}^{'} := (P,\mathcal{C},\mathcal{W}^{'})$
\Statex\textbf{Return} $\mathcal{G}^{'}$
\end{algorithmic}
\end{algorithm}

%%%%%%%%%%%%%%%%%%%%%%%%%%%%%%%%%%%%%%%%%%%%%%%%%%%%%%%%%%%%%%%%%%%%%%%%

\end{document}